\documentclass[11pt,fleqn]{article}

\usepackage{amsmath,amssymb,amsthm,enumerate,cite}%,backref

\newcommand{\p}{\partial}
\newcommand{\const}{\mathop{\rm const}\nolimits}
\newcommand{\Equiv}{\mathop{ \sim}}

\newcommand{\CV}{\mathop{\rm CV}\nolimits}
\newcommand{\CL}{\mathop{\rm CL}\nolimits}
\newcommand{\Ch}{\mathop{\rm Ch}\nolimits}
\newcommand{\Eop}{\mathop{\sf E}}
\newcommand{\Fder}{\mathop{\sf D}}
\newcommand{\prho}{p}

\newcounter{tbn}

\newcounter{mcasenum}
\renewcommand{\themcasenum}{{\rm\arabic{mcasenum}}}

\newtheorem{theorem}{Theorem}
\newtheorem{lemma}{Lemma}
\newtheorem{corollary}{Corollary}
\newtheorem{proposition}{Proposition}
{\theoremstyle{definition} \newtheorem{definition}{Definition}
\newtheorem{example}{Example}

\newtheorem{note}{Note}
\begin{document}

\par\noindent {\LARGE\bf
Group Analysis of Variable Coefficient\\ Diffusion--Convection Equations.\\ III. Conservation Laws %and Contractions
\par}
{\vspace{4mm}\par\noindent {\bf N.M. Ivanova~$^\dag$, R.O. Popovych~$^\ddag$ and C. Sophocleous~$^\S$
} \par\vspace{2mm}\par}
{\vspace{2mm}\par\noindent {\it
$^\dag{}^\ddag$~Institute of Mathematics of NAS of Ukraine, %\\
%$\phantom{^\dag{}^\ddag}{}$~
3 Tereshchenkivska Str., 01601 Kyiv, Ukraine\\
}}
{\noindent \vspace{2mm}{\it
$\phantom{^\dag{}^\ddag}$~e-mail: ivanova@imath.kiev.ua, rop@imath.kiev.ua
}\par}

{\par\noindent\vspace{2mm} {\it
$^\ddag$~Fakult\"at f\"ur Mathematik, Universit\"at Wien, Nordbergstra{\ss}e 15, A-1090 Wien, Austria
} \par}

{\vspace{2mm}\par\noindent {\it
$^\S$~Department of Mathematics and Statistics, University of Cyprus,
CY 1678 Nicosia, Cyprus\\
}}
{\noindent {\it
$\phantom{^\S}$~e-mail: christod@ucy.ac.cy
} \par}

{\vspace{5mm}\par\noindent\hspace*{8mm}\parbox{140mm}{\small
The notions of generating sets of conservation laws of systems of differential equations with respect to symmetry groups and equivalence groups
are introduced and applied. This allows us to generalize essentially the procedure of finding potential symmetries
for the systems with multidimensional spaces of conservation laws.
A class of variable coefficient
(1+1)-dimensional nonlinear diffusion--convection equations of general form
$f(x)u_t=(g(x)A(u)u_x)_x+h(x)B(u)u_x$ is investigated.
Using the most direct method, we carry out two classifications of local conservation
laws up to equivalence relations generated by both usual and enhanced equivalence groups.
Equivalence with respect to~$\hat G^{\sim}$ and correct choice of gauge coefficients of equations
play the major role for simple and clear formulation of the final results.
%As an example, we study an interesting equation from the class under consideration
%with symmetry point of view in more detail.
The notion of contractions of conservation laws and one of characteristics of conservation laws are introduced and
contractions of conservation laws of diffusion--convection equations are found.
}\par\vspace{5mm}}

%!!!!!!

%What is POTENTIAL TRANSFORMATIONS?

%For system = For system for some values of parameters?

%Should cases 5b and 5c be united?

\section{Introduction}

The presented paper continues the series of works on modern trends of group analysis
%started by papers~\cite{Ivanova&Popovych&Sophocleous2006Part1,Ivanova&Popovych&Sophocleous2006Part2} and
illustrated on example of nonlinear variable coefficient diffusion--convection equations
\begin{equation} \label{eqDKfgh}
f(x)u_t=(g(x)A(u)u_x)_x+h(x)B(u)u_x
\end{equation}
started in~\cite{Ivanova&Popovych&Sophocleous2006Part1,Ivanova&Popovych&Sophocleous2006Part2}.
Here $f=f(x),$ $g=g(x),$ $h=h(x),$ $A=A(u)$ and $B=B(u)$ are arbitrary smooth functions of their variables,
$f(x)g(x)A(u)\!\neq\! 0$.

In the first two parts of the presented series we have shown an importance of detailed study of different kinds
of equivalence transformations for solving the group classification problem and some its applications.
In particular, we investigated equivalence groups (usual and extended ones), discussed their structure
and performed the complete group classification of class~\eqref{eqDKfgh}.
We determine that the complete usual equivalence group~$G^{\sim}$ of class~\eqref{eqDKfgh} consists of the transformations
\begin{gather*}
\tilde t=\delta_1 t+\delta_2,\quad
\tilde x=X(x), \quad
\tilde u=\delta_3 u+\delta_4, \\
\tilde f=\dfrac{\varepsilon_1\delta_1}{X_x} f, \quad
\tilde g=\varepsilon_1\varepsilon_2^{-1}X_x\, g, \quad
\tilde h=\varepsilon_1\varepsilon_3^{-1}h, \quad
\tilde A=\varepsilon_2A, \quad
\tilde B=\varepsilon_3B,
\end{gather*}
where $\delta_j$ $(j=\overline{1,4})$ and $\varepsilon_i$ $(i=\overline{1,3})$ are arbitrary constants,
$\delta_1\delta_3\varepsilon_1\varepsilon_2\varepsilon_3\not=0$, $X$ is an arbitrary smooth function of~$x$, $X_x\not=0$.

It appears that class~\eqref{eqDKfgh} admits further equivalence transformations which do not belong to~$G^{\sim}$
and depend on arbitrary elements in some non-fixed (possibly, nonlocal) way.
We have constructed the complete in this sense extended equivalence group~$\hat G^{\sim}$ of
class~\eqref{eqDKfgh}, using the direct method.
$\hat G^{\sim}$ is formed by the transformations
\begin{gather*}
\tilde t=\delta_1 t+\delta_2,\quad
\tilde x=X(x), \quad
\tilde u=\delta_3 u+\delta_4, \\
\tilde f=\dfrac{\varepsilon_1\delta_1\varphi}{X_x}f, \quad
\tilde g=\varepsilon_1\varepsilon_2^{-1}X_x\varphi\,g, \quad
\tilde h=\varepsilon_1\varepsilon_3^{-1}\varphi\,h, \quad
\tilde A=\varepsilon_2 A, \quad
\tilde B=\varepsilon_3 (B+\varepsilon_4 A),
\end{gather*}
where $\delta_j$ $(j=\overline{1,4})$ and $\varepsilon_i$ $(i=\overline{1,4})$ are arbitrary constants,
$\delta_1\delta_3\varepsilon_1\varepsilon_2\varepsilon_3\not=0$,
$X$ is an arbitrary smooth function of~$x$, $X_x\not=0$,
$\varphi=e^{-\varepsilon_4\int \frac{h(x)}{g(x)}dx}$.

It is shown that application of extended equivalence group~$\hat G^{\sim}$
is of vital importance for obtaining group classification in closed
explicit form~\cite{Ivanova&Popovych&Sophocleous2006Part1,Ivanova&Popovych&Sophocleous2004}.
Likewise we state below the necessity of application of equivalence transformations to investigation of conservation laws.
More precisely, it will be shown that classification of conservation laws of equations~\eqref{eqDKfgh}
with respect to the usual equivalence group $G^{\sim}$ can be formulated in an implicit form only.
At the same time, using the extended equivalence group~$\hat G^{\sim}$,
we can present the result of classification in a closed and simple form with a smaller number
of inequivalent equations having nontrivial conservation laws.

Using the gauge equivalence transformation $\tilde t=t$, $\tilde x=\int \frac{dx}{g(x)}$, $\tilde u=u$ we can reduce equation~(\ref{eqDKfgh}) to
\[
\tilde f(\tilde x)\tilde u_{\tilde t}= (A(\tilde u)
\tilde u_{\tilde x})_{\tilde x} + \tilde h(\tilde x)B(\tilde u)\tilde u_{\tilde x},
\]
where $\tilde f(\tilde x)=g(x)f(x)$, $\tilde g(\tilde x)=1$ and $\tilde h(\tilde x)=h(x)$.
%(Likewise any equation of form~\eqref{eqDKfgh} can be reduced to the same form with $\tilde f(\tilde x)=1.$)
That is why, without loss of generality we can restrict ourselves to investigation of the equation
\begin{equation} \label{eqDKfh}
f(x)u_t=\left(A(u)u_x \right)_x + h(x)B(u)u_x.
\end{equation}

Any transformation from~$\hat G^{\sim}$, which preserves the condition $g = 1$, has the form
\begin{equation} \label{EquivTransformationsDKfh}\arraycolsep=0em
\begin{array}{l}
\tilde t=\delta_1 t+\delta_2,\quad
\tilde x=\delta_5 \int e^{\delta_8\int\! h}dx+\delta_6, \quad
\tilde u=\delta_3 u+\delta_4,\\[1ex]
\tilde f=\delta_1\delta_5^{-1}\delta_9 fe^{-2\delta_8\int\! h}, \quad
\tilde h=\delta_9\delta_7^{-1} he^{-\delta_8\int\! h}, \\[1ex]
\tilde A=\delta_5\delta_9A, \quad
\tilde B=\delta_7(B+\delta_8A),
\end{array}
\end{equation}
where $\delta_i$ ($i=\overline{1,9}$) are arbitrary constants, $\delta_1\delta_3\delta_5\delta_7\delta_9\not=0$.
(Here and below $\int\! h=\int\! h(x)\,dx$.)
The set~$\hat G^{\sim}_1$ of such transformations
is a subgroup of~$\hat G^{\sim}$.

Result of group classification allows us to construct a number of exact solutions of different equations from class~\eqref{eqDKfgh}.
In the second part~\cite{Ivanova&Popovych&Sophocleous2006Part2} of the series we have used two different approaches for finding exact solutions:
direct application of Lie (conditional and generalized conditional) symmetries
and reconstruction of new solutions by acting on the known ones with additional equivalence transformations.

After analyzing the obtained results it appears that some of the previously classified equations can be regarded as limiting cases of
the other ones, that leads to the natural notion of contractions of (systems of) equations introduced
in the second part~\cite{Ivanova&Popovych&Sophocleous2006Part2} of the series. Ibid we investigated contractions of equations from class~\eqref{eqDKfgh}.

These ideas can be generalized for investigation of conservation laws of systems of differential equations.
Thus, e.g., there exists a close similarity between characteristics of conservation laws and symmetries that leads to the natural notions
of contractions of characteristics and conservation laws that will be considered below. These notions are illustrated by
example of conservation laws of equations~\eqref{eqDKfgh}.

Similarly to the group classification problem, we show that the complete description of the spaces of conservation laws
of equations~\eqref{eqDKfgh} is impossible without essential usage of (usual and extended) equivalence transformations.
Investigation of such equivalences allows us to simplify essentially technical calculations and to obtain results
in closed explicit form.
Equivalence with respect to the extended equivalence group and correct choice of gauge coefficients of equations
allow us to obtain clear formulation of the final results.
Note that for wide classes of nonlinear equations, this is the only way
to obtain the complete description of spaces of conservation laws
(see, e.g.,~\cite{Ivanova&Popovych&Sophocleous2004,Ivanova2006CLDifMultiDim,Popovych&Ivanova2004ConsLawsLanl}
for detailed analysis of application of such equivalences
for finding conservation laws of different classes of diffusion-type equations).

The notion of equivalence of conservation laws with respect to a group of transformations
which was introduced in~\cite{Popovych&Ivanova2004ConsLawsLanl} can be generalized in several directions:
classification of pairs ``system + space of conservation laws'';
classification of conservation laws for a given system with respect to its symmetry group;
classification of pairs ``system + a conservation law''.
Such classification schemes are helpful both for ordering and deeper understanding of the found conservation laws and
for further applications.
For instance, investigation of different generating sets of conservation laws of equations~\eqref{eqDKfgh} allows us to
close the ``blank spot'' in theory of potential symmetries.
Namely, previously, for construction of simplest potential systems in cases when the dimension of the space of local conservation laws
is greater then one,  only basis conservation laws were used.
However, the basis conservation laws may be equivalent with respect to groups of symmetry transformations,
or vice versa, the number of $G^{\sim}$-independent linear combinations of conservation laws
may be grater then dimension of the space of conservation laws.
The first possibility leads to an unnecessary, often cumbersome, investigation of equivalent systems,
the second one makes possible missing a great number of inequivalent potential systems.
Below we show how to choose conservation laws in order to obtain all possible inequivalent potential systems
associated to the given system.

In Section~\ref{SectionBasicDef} we adduce basic definitions and statements on conservation laws.
An important notion of characteristics of conservation laws is discussed in Section~\ref{SectionOnCharacteristicsOfConsLaws}.
Following the spirit of~\cite{Popovych&Ivanova2004ConsLawsLanl} we repeat a definition of equivalence of conservation laws
with respect to a transformation group in Section~\ref{SectionOnEquivOfConsLaws}.
In Section~\ref{SectionOnEquivOfConsLaws} we introduce also the notions of generating sets of conservation laws
with respect to symmetry groups and with respect to equivalence groups.
After that (Section~\ref{SectionOnConsLawsOfEqDKfgh}) two different classifications of conservation laws of equations~\eqref{eqDKfgh}
are presented, namely, classification with respect to the usual equivalence group and one with respect to the extended equivalence group.
Next (Section~\ref{SectionOnContractionsOfCLsOfDKfgh}), we introduce the notions of contractions conservation laws
and ones of characteristics of conservation laws and give some examples of the contractions in class~\eqref{eqDKfgh}.
Generating sets of conservation laws of diffusion--convection equations~\eqref{eqDKfgh} are investigated in Section~\ref{SectionOnGenSetsOfCLsDCEs}.
Such investigation form a basis for finding all possible inequivalent potential systems of equations~\eqref{eqDKfgh}
presented in Section~\ref{SectionPotSysOfDifConvEq}.
In Section~\ref{SectionOnPotSymPotCLsTheory} we discuss shortly potential symmetries and potential conservation laws in general case
and in case of systems with two independent variables.
Potential conservation laws of equations~\eqref{eqDKfgh} are investigated in Section~\ref{SectionOnPotConsLawsOfDCEs}.

A detailed investigation of potential symmetries obtained from the constructed potential systems will be performed in the last
part~\cite{Ivanova&Popovych&Sophocleous2006Part4} of this series.

\section{Basic definitions and statements on conservation laws}\label{SectionBasicDef}

In this and next two sections we give basic definitions and statements on conservation laws and
formulate the notion of equivalence of conservation laws with respect to equivalence groups,
which was first introduced in~\cite{Popovych&Ivanova2004ConsLawsLanl} and some of its generalizations.
This notion is a basis for modification of the direct method of construction of conservation laws,
which is applied in Section~\ref{SectionOnConsLawsOfEqDKfgh} for exhaustive classification of
local conservation laws equations from class~(\ref{eqDKfgh}).

Let~$\mathcal{L}$ be a system~$L(x,u_{(\rho)})=0$ of $l$ differential equations $L^1=0$, \ldots, $L^l=0$
for $m$ unknown functions $u=(u^1,\ldots,u^m)$
of $n$ independent variables $x=(x_1,\ldots,x_n).$
Here $u_{(\rho)}$ denotes the set of all the derivatives of the functions $u$ with respect to $x$
of order no greater than~$\rho$, including $u$ as the derivatives of the zero order.
Let $\mathcal{L}_{(k)}$ denote the set of all algebraically independent differential consequences
that have, as differential equations, orders not greater than $k$. We identify~$\mathcal{L}_{(k)}$ with
the manifold determined by~$\mathcal{L}_{(k)}$ in the jet space~$J^{(k)}$.

\begin{definition}\label{def.conservation.law}
A {\em conserved vector} of the system~$\mathcal{L}$ is
an $n$-tuple $F=(F^1(x,u_{(r)}),\ldots,F^n(x,u_{(r)}))$ for which the divergence ${\rm Div}\,F=D_iF^i$
vanishes for all solutions of~$\mathcal{L}$ (i.e. ${\rm Div}F\bigl|_\mathcal{L}=0$).
\end{definition}

In Definition~\ref{def.conservation.law} and below
$D_i=D_{x_i}$ denotes the operator of total differentiation with respect to the variable~$x_i$, i.e.,
$D_i=\p_{x_i}+u^a_{\alpha,i}\p_{u^a_\alpha}$, where
$u^a_\alpha$ and $u^a_{\alpha,i}$ stand for the variables in jet spaces,
which correspond to derivatives
$\p^{|\alpha|}u^a/\p x_1^{\alpha_1}\ldots\p x_n^{\alpha_n}$ and $\p u^a_\alpha/\p x_i$,
$\alpha=(\alpha_1,\ldots,\alpha_n)$,
$\alpha_i\in\mathbb{N}\cup\{0\}$, $|\alpha|{:}=\alpha_1+\cdots+\alpha_n$.
We use the summation convention for repeated indices and assume any function as its zero-order derivative.
The notation~$V\bigl|_\mathcal{L}$ means that values of $V$ are considered
only on solutions of the system~$\mathcal{L}$.

\begin{definition}
A conserved vector $F$ is called {\em trivial} if $F^i=\hat F^i+\check F^i,$ $i=\overline{1,n},$
where $\hat F^i$ and $\check F^i$ are, likewise $F^i$, functions of $x$ and derivatives of $u$
(i.e., differential functions),
$\hat F^i$ vanish on the solutions of~$\mathcal L$ and the $n$-tuple $\check F=(\check F^1,\ldots,\check F^n)$
is a null divergence (i.e., its divergence vanishes identically).
\end{definition}

The triviality concerning the vanishing conserved vectors on solutions of the system can be easily
eliminated by confining on the manifold of the system, taking into account all its necessary differential consequences.
A characterization of all null divergences is given by the following lemma (see e.g.~\cite{Olver1986}).

\begin{lemma}\label{lemma.null.divergence}
The $n$-tuple $F=(F^1,\ldots,F^n)$, $n\ge2$, is a null divergence ($\mathop{\rm Div}\nolimits F\equiv0$)
iff there exist smooth functions $v^{ij}$ ($i,j=\overline{1,n}$) of $x$ and derivatives of $u$,
such that $v^{ij}=-v^{ji}$ and $F^i=D_jv^{ij}$.
\end{lemma}

The functions $v^{ij}$ are called {\em potentials} corresponding to the null divergence~$F$.
If $n=1$ any null divergence is constant.

\begin{definition}\label{DefinitionOfConsVectorEquivalence}
Two conserved vectors $F$ and $F'$ are called {\em equivalent} if
the vector-function $F'-F$ is a trivial conserved vector.
\end{definition}

The above definitions of triviality and equivalence of conserved vectors are natural
in view of the usual ``empiric'' definition of conservation laws of a system of differential equations
as divergences of its conserved vectors, i.e., divergence expressions which vanish for all solutions of this system.
For example, equivalent conserved vectors correspond to the same conservation law.
It allows us to formulate the definition of conservation law in a rigorous style (see, e.g.,~\cite{Zharinov1986}).
Namely, for any system~$\mathcal{L}$ of differential equations the set~$\CV(\mathcal{L})$ of its conserved vectors is a linear space,
and the subset~$\CV_0(\mathcal{L})$ of trivial conserved vectors is a linear subspace in~$\CV(\mathcal{L})$.
The factor space~$\CL(\mathcal{L})=\CV(\mathcal{L})/\CV_0(\mathcal{L})$
coincides with the set of equivalence classes of~$\CV(\mathcal{L})$ with respect to the equivalence relation adduced in
Definition~\ref{DefinitionOfConsVectorEquivalence}.

\begin{definition}\label{DefinitionOfConsLaws}
The elements of~$\CL(\mathcal{L})$ are called {\em conservation laws} of the system~$\mathcal{L}$,
and the whole factor space~$\CL(\mathcal{L})$ is called {\em the space of conservation laws} of~$\mathcal{L}$.
\end{definition}

That is why we assume description of the set of conservation laws
as finding~$\CL(\mathcal{L})$ which is equivalent to construction of either a basis if
$\dim \CL(\mathcal{L})<\infty$ or a system of generatrices in the infinite dimensional case.
The elements of~$\CV(\mathcal{L})$ which belong to the same equivalence class giving a conservation law~${\cal F}$
are considered all as conserved vectors of this conservation law,
and we will additionally identify elements from~$\CL(\mathcal{L})$ with their representatives
in~$\CV(\mathcal{L})$.
For $F\in\CV(\mathcal{L})$ and ${\cal F}\in\CL(\mathcal{L})$
the notation~$F\in {\cal F}$ will denote that $F$ is a conserved vector corresponding
to the conservation law~${\cal F}$.
In contrast to the order $r_F$ of a conserved vector~$F$ as the maximal order of derivatives explicitly appearing in~$F$,
the {\em order of the conservation law}~$\cal F$
is called $\min\{r_F\,|\,F\in{\cal F}\}$.
Under linear dependence of conservation laws we understand linear dependence of them as elements of~$\CL(\mathcal{L})$.
Therefore, in the framework of ``representative'' approach
conservation laws of a system~$\mathcal{L}$ are considered {\em linearly dependent} if
there exists a linear combination of their representatives, which is a trivial conserved vector.

\section{Characteristics of conservation laws}\label{SectionOnCharacteristicsOfConsLaws}

Let the system~$\cal L$ be totally nondegenerate~\cite{Olver1986}.
Then application of the Hadamard lemma to the definition of conserved vector and integrating by parts imply that
divergence of any conserved vector of~$\mathcal L$ can be always presented,
up to the equivalence relation of conserved vectors,
as a linear combination of left hand sides of independent equations from $\mathcal L$
with coefficients~$\lambda^\mu$ being functions on a suitable jet space~$J^{(k)}$:
\begin{equation}\label{CharFormOfConsLaw}
\mathop{\rm Div}\nolimits F=\lambda^\mu L^\mu.
\end{equation}
Here the order~$k$ is determined by~$\mathcal L$ and the allowable order of conservation laws,
$\mu=\overline{1,l}$.

\begin{definition}\label{DefCharForm}
Formula~\eqref{CharFormOfConsLaw} and the $l$-tuple $\lambda=(\lambda^1,\ldots,\lambda^l)$
are called the {\it characteristic form} and the {\it characteristic}
of the conservation law~$\mathop{\rm Div}\nolimits F=0$ correspondingly.
\end{definition}

The characteristic~$\lambda$ is {\em trivial} if it vanishes for all solutions of $\cal L$.
Since $\cal L$ is nondegenerate, the characteristics~$\lambda$ and~$\tilde\lambda$ satisfy~\eqref{CharFormOfConsLaw}
for the same~$F$ and, therefore, are called {\em equivalent}
iff $\lambda-\tilde\lambda$ is a trivial characteristic.
Similarly to conserved vectors, the set~$\Ch(\mathcal{L})$ of characteristics
corresponding to conservation laws of the system~$\cal L$ is a linear space,
and the subset~$\Ch_0(\mathcal{L})$ of trivial characteristics is a linear subspace in~$\Ch(\mathcal{L})$.
The factor space~$\Ch_{\rm f}(\mathcal{L})=\Ch(\mathcal{L})/\Ch_0(\mathcal{L})$
coincides with the set of equivalence classes of~$\Ch(\mathcal{L})$
with respect to the above characteristic equivalence relation.

The following result~\cite{Olver1986} forms the cornerstone for the methods of studying conservation laws,
which are based on formula~\eqref{CharFormOfConsLaw}, including the Noether theorem and
the direct method in the version by
Anco and Bluman~\cite{Anco&Bluman2002a,Anco&Bluman2002b}.

\begin{theorem}[\cite{Olver1986}]\label{TheoremIsomorphismChCV}
Let~$\mathcal{L}$ be a normal, totally nondegenerate system of differential equations.
Then representation of conservation laws of~$\mathcal{L}$ in the characteristic form~\eqref{CharFormOfConsLaw}
generates a one-to-one linear mapping between~$\CL(\mathcal{L})$ and~$\Ch_{\rm f}(\mathcal{L})$.
\end{theorem}

Using properties of total divergences, we can exclude the conserved vector~$F$ from~\eqref{CharFormOfConsLaw}
and obtain a condition for the characteristic~$\lambda$ only.
Namely, a differential function~$f$ is a total divergence, i.e., $f=\mathop{\rm Div} F$
for some $n$-tuple~$F$ of differential functions iff $\Eop(f)=0$.
Hereafter the Euler operator~$\Eop=(\Eop^1,\ldots, \Eop^m)$ is the $m$-tuple of differential operators
\[
{\Eop}^a=(-D)^\alpha\p_{u^a_\alpha}, \quad a=\overline{1,m},
\]
where
$\alpha=(\alpha_1,\ldots,\alpha_n)$ runs the multi-indices set ($\alpha_i\!\in\!\mathbb{N}\cup\{0\}$),
$(-D)^\alpha=(-D_1)^{\alpha_1}\ldots(-D_m)^{\alpha_m}$.
Therefore, action of the Euler operator on~\eqref{CharFormOfConsLaw}
results to the equation
\begin{equation}\label{NSCondOnChar}
\Eop(\lambda^\mu L^\mu)={\Fder}_\lambda^*(L)+{\Fder}_L^*(\lambda)=0,
\end{equation}
which is a necessary and sufficient condition on characteristics of conservation laws for the system~$\mathcal{L}$.
The matrix differential operators~${\Fder}_\lambda^*$ and~${\Fder}_L^*$ are the adjoints of
the Fr\'echet derivatives~${\Fder}_\lambda^{\phantom{*}}$ and~${\Fder}_L^{\phantom{*}}$, i.e.,
\[
{\Fder}_\lambda^*(L)=\left((-D)^\alpha\left( \dfrac{\p\lambda^\mu}{\p u^a_\alpha}L^\mu\right)\right), \qquad
{\Fder}_L^*(\lambda)=\left((-D)^\alpha\left( \dfrac{\p L^\mu}{\p u^a_\alpha}\lambda^\mu\right)\right).
\]
Since ${\Fder}_\lambda^*(L)=0$ automatically on solutions of~$\mathcal{L}$ then
equation~\eqref{NSCondOnChar} implies a necessary condition for $\lambda$ to belong to~$\Ch(\mathcal{L})$:
\begin{equation}\label{NCondOnChar}
{\Fder}_L^*(\lambda)\bigl|_{\mathcal{L}}=0.
\end{equation}
Condition~\eqref{NCondOnChar} can be considered as adjoint to the criterion
${\Fder}_L^{\phantom{*}}(\eta)\bigl|_{\mathcal{L}}=0$ for infinitesimal invariance of $\mathcal{L}$
with respect to evolutionary vector field having the characteristic~$\eta=(\eta^1,\ldots,\eta^m)$.
That is why solutions of~\eqref{NCondOnChar} are called sometimes as
{\em cosymmetries}~\cite{Blaszak1998} or
{\em adjoint symmetries}~\cite{Anco&Bluman2002b}.

\section{Equivalence of conservation laws\\ with respect to transformation groups}\label{SectionOnEquivOfConsLaws}

We can essentially simplify and order classification of conservation laws, taking into account additionally
symmetry transformations of a system or equivalence transformations of a whole class of systems.
Such problem is similar to one of group classification of differential equations.

\begin{proposition}
Any point transformation~$g$ maps a class of equations in the conserved form into itself.
More exactly, the transformation~$g$: $\tilde x=x_g(x,u)$, $\tilde u=u_g(x,u)$ prolonged to the jet space~$J^{(r)}$
transforms the equation $D_iF^i=0$ to the equation $D_iF^i_g=0$. The transformed conserved vector~$F_g$ is determined
by the formula
\begin{equation}\label{eq.tr.var.cons.law}
F_g^i(\tilde x,\tilde u_{(r)})=\frac{D_{x_j}\tilde x_i}{|D_x\tilde x|}\,F^j(x,u_{(r)}),
\quad\mbox{i.e.}\quad
F_g(\tilde x,\tilde u_{(r)})=\frac{1}{|D_x\tilde x|}(D_x\tilde x)F(x,u_{(r)})
\end{equation}
in the matrix notions. Here $|D_x\tilde x|$ is the determinant of the matrix $D_x\tilde x=(D_{x_j}\tilde x_i)$.
\end{proposition}

\begin{note}
In the case of one dependent variable ($m=1$) $g$ can be a contact transformation:
$\tilde x=x_g(x,u_{(1)})$, $\tilde u_{(1)}=u_{g(1)}(x,u_{(1)})$.
Similar notes are also true for the statements below.
\end{note}

\begin{definition}
Let $G$ be a symmetry group of the system~$\mathcal{L}$.
Two conservation laws with the conserved vectors $F$ and $F'$ are called {\em $G$-equivalent} if
there exists a transformation $g\in G$ such that the conserved vectors $F_g$ and $F'$
are equivalent in the sense of Definition~\ref{DefinitionOfConsVectorEquivalence}.
\end{definition}

Any transformation $g\in G$ induces a linear one-to-one mapping $g_*$ in~$\CV(\mathcal{L})$,
transforms trivial conserved vectors only to trivial ones
(i.e., $\CV_0(\mathcal{L})$ is invariant with respect to~$g_*$)
and therefore, induces a linear one-to-one mapping $g_{\rm f}$ in~$\CL(\mathcal{L})$.
It is obvious that $g_{\rm f}$ preserves linear (in)dependence of elements
in~$\CL(\mathcal{L})$ and maps a basis (a set of generatrices) of~$\CL(\mathcal{L})$
in a basis (a set of generatrices) of the same space.
In such way we can consider the $G$-equivalence relation of conservation laws
as well-determined on~$\CL(\mathcal{L})$ and use it to classify conservation laws.

\begin{proposition}
If the system~$\mathcal{L}$ admits a one-parameter group of transformations then the infinitesimal generator
$X=\xi^i\p_i+\eta^a\p_{u^a}$
of this group can be used for construction of new conservation laws from known ones.
Namely, differentiating equation~(\ref{eq.tr.var.cons.law})
with respect to the parameter $\varepsilon$ and taking the value $\varepsilon=0$,
we obtain the new conserved vector
\begin{equation}\label{eq.inf.tr.var.cons.law}
\widetilde F^i=-X_{(r)}F^i+(D_j\xi^i)F^j-(D_j\xi^j)F^i.
\end{equation}
Here $X_{(r)}$ denotes the $r$-th prolongation~\cite{Olver1986,Ovsiannikov1982} of the operator $X$.
\end{proposition}

\begin{note}Formula~\eqref{eq.inf.tr.var.cons.law} can be directly extended to generalized symmetry operators
(see, for example,~\cite{Bluman&Anco2002,Kara&Mahomed2002}).
A similar statement for generalized symmetry operators in evolutionary form ($\xi^i=0$)
was known earlier~\cite{Ibragimov1985,Olver1986}.
It was used in~\cite{Khamitova1982} to introduce a notion of basis of conservation laws as a set
which generates a whole set of conservation laws with action of generalized symmetry operators and operation
of linear combination.
\end{note}

\begin{proposition}\label{PropositionOnInducedMapping}
Any point transformation $g$ between systems~$\mathcal{L}$ and~$\tilde{\mathcal{L}}$
induces a one-to-one linear mapping $g_*$ from~$\CV(\mathcal{L})$ into~$\CV(\tilde{\mathcal{L}})$,
which maps $\CV_0(\mathcal{L})$ into~$\CV_0(\tilde{\mathcal{L}})$
and generates a one-to-one linear mapping $g_{\rm f}$ from~$\CL(\mathcal{L})$ to~$\CL(\tilde{\mathcal{L}})$.
\end{proposition}

\begin{corollary}\label{CorollaryOnInducedMappingOfChar}
Any point transformation $g$ between systems~$\mathcal{L}$ and~$\tilde{\mathcal{L}}$
induces a one-to-one linear mapping $\hat g_{\rm f}$ from~$\Ch_{\rm f}(\mathcal{L})$
to~$\Ch_{\rm f}(\tilde{\mathcal{L}})$.
\end{corollary}
It is possible to obtain an explicit formula for correspondence between characteristics of~$\mathcal{L}$
and~$\tilde{\mathcal{L}}$.
Let $\tilde{\mathcal{L}}^\mu=\Lambda^{\mu\nu}\mathcal{L}^\nu$,
where $\Lambda^{\mu\nu}=\Lambda^{\mu\nu\alpha}D^\alpha$, $\Lambda^{\mu\nu\alpha}$ are differential functions,
$\alpha=(\alpha_1,\ldots,\alpha_n)$ runs the multi-indices set ($\alpha_i\!\in\!\mathbb{N}\cup\{0\}$),
$\mu,\nu=\overline{1,l}$.
Then
\[\lambda^\mu={\Lambda^{\nu\mu}}^*(|D_x\tilde x|\tilde\lambda^\nu).\]
Here ${\Lambda^{\nu\mu}}^*=(-D)^\alpha\cdot\Lambda^{\mu\nu\alpha}$ is the adjoint to the operator~$\Lambda^{\nu\mu}$.
For a number of cases, e.g. if~$\mathcal{L}$ and~$\tilde{\mathcal{L}}$ are single partial differential equations
($l=1$), the operators~$\Lambda^{\mu\nu}$ are simply differential functions
(i.e. $\Lambda^{\mu\nu\alpha}=0$ for $|\alpha|>0$) and, therefore, ${\Lambda^{\nu\mu}}^*=\Lambda^{\mu\nu}$.

Consider the class~$\mathcal{L}|_{\cal S}$ of systems~$\mathcal{L}_\theta$: $L(x,u_{(\rho)},\theta(x,u_{(\rho)}))=0$
parameterized with the parameter-functions~$\theta=\theta(x,u_{(\rho)}).$
Here $L$ is a tuple of fixed functions of $x,$ $u_{(\rho)}$ and $\theta.$
$\theta$~denotes the tuple of arbitrary (parametric) functions
$\theta(x,u_{(\rho)})=(\theta^1(x,u_{(\rho)}),\ldots,\theta^k(x,u_{(\rho)}))$
running the set~${\cal S}$ of solutions of the system~$S(x,u_{(\rho)},\theta_{(q)}(x,u_{(\rho)}))=0$.
This system consists of differential equations on $\theta$,
where $x$ and $u_{(\rho)}$ play the role of independent variables
and $\theta_{(q)}$ stands for the set of all the partial derivatives of $\theta$ of order no greater than $q$.
In what follows we call the functions $\theta$ arbitrary elements.
Sometimes the set $\mathcal{S}$ is additionally constrained by the non-vanish condition
$S'(x,u_{( p)},\theta_{(q)}(x,u_{( p)}))\ne0$ with another tuple $S'$ of differential functions.
Denote the point transformations group preserving the
form of the systems from~$\mathcal{L}|_{\cal S}$ as $G^{\Equiv}=G^{\Equiv}(L,S).$

Consider the set~$P=P(L,S)$ of all pairs each of which consists of
a system $\mathcal{L}_\theta$ from~$\mathcal{L}|_{\cal S}$ and a~conservation law~${\cal F}$ of this system.
In view of Proposition~\ref{PropositionOnInducedMapping},
action of symmetry transformations of the system~$\mathcal{L}_\theta$ and transformations from~$G^{\Equiv}$ on $\mathcal{L}|_{\cal S}$ and
$\{\CV(\mathcal{L}_{\theta})\,|\,\theta\in{\cal S}\}$
together with the pure equivalence relation of conserved vectors
naturally generates equivalence relations on~$P$.
Such generations can be made in several directions:
classification of pairs ``system + space of conservation laws'';
classification of conservation laws for a given system with respect to its symmetry group;
classification of pairs ``system + a conservation law''.
Let us consider these possibilities in more detail.

\medskip

\noindent{\bf 1. Equivalence with respect to an equivalence group~\cite{Popovych&Ivanova2004ConsLawsLanl}.}
We wish to find a list of systems from the class~$\mathcal{L}|_{\cal S}$ where
(i) $\mathcal{L}_\theta$ are graded by the dimensions of conservation laws and
(ii) we can choose a representative, all other systems are equivalent to,
or we can write conditions that systems belong to a given class.

\begin{definition}
Let $\theta,\theta'\in{\cal S}$,
${\cal F}\in\CL(\mathcal{L}_\theta)$, ${\cal F}'\in\CL(\mathcal{L}_{\theta'})$,
$F\in{\cal F}$, $F'\in{\cal F'}$.
The pairs~$(\mathcal{L}_\theta,{\cal F})$ and~$(\mathcal{L}_{\theta'},{\cal F'})$
are called {\em $G^{\Equiv}$-equivalent} if there exists a transformation $g\in G^{\Equiv}$
which transform the system~$\mathcal{L}_\theta$ to the system~$\mathcal{L}_{\theta'}$ and
such that the conserved vectors $F_g$ and $F'$
are equivalent in the sense of Definition~\ref{DefinitionOfConsVectorEquivalence}.
\end{definition}

Classification of conservation laws with respect to~$G^{\Equiv}$ will be understood as
classification in~$P$ with respect to the above equivalence relation.
This problem can be investigated in a way that is similar to group classification in classes
of systems of differential equations, especially if it is formulated in terms of characteristics.
Namely, we construct firstly the conservation laws that are defined for all values of the arbitrary elements.
(The corresponding conserved vectors may depend on the arbitrary elements.)
Then we classify, with respect to the equivalence group, arbitrary elements for each of that the system
admits additional conservation laws.

In an analogues way we also can introduce equivalence relations on~$P$, which are
generated by either generalizations of usual equivalence groups or
all admissible point or contact transformations
(called also form-preserving ones~\cite{Kingston&Sophocleous1998})
in pairs of equations from~$\mathcal{L}|_{\cal S}$.

\medskip

\noindent{\bf 2. Generating sets of conservation laws.}
Consider system~$\mathcal{L}_\theta$ and a non-empty subset $\mathfrak F=\{F_\gamma,\gamma\in\Gamma\}$ of the set $\CV(\mathcal{L}_\theta)$
of its conserved vectors.
Acting on $F_\gamma$ by the transformations from the symmetry group~$G^{\max}(\mathcal{L}_\theta)$ of~$\mathcal{L}_\theta$
and taking linear combinations, we may obtain wider set of conserved vectors.
Roughly speaking, the set~$\mathfrak F$ is independent with respect to $G^{\max}$ if an arbitrary nonzero linear combination
of $G^{\max}$-transformed conservation laws is  nontrivial. The rigorous definition is as follows.

\begin{definition}
The subset $\mathfrak F=\{F_\gamma,\gamma\in\Gamma\}\subset\CV(\mathcal{L}_\theta)$ is {\em dependent
with respect to the group~$G^{\max}(\mathcal{L}_\theta)$ of symmetry transformations} if
\[
\forall\{\gamma_1,\ldots,\gamma_n\}\subset\Gamma,\ \forall g^i\in G^{\max}(\mathcal{L}_\theta)\quad
\lambda_ig^i_*F_\gamma=0\Rightarrow \lambda_i=0,\ i=\overline{1,n}.
\]
\end{definition}

The set of $G^{\max}(\mathcal{L}_\theta)$-independent conserved vectors generating the whole space~$\CV(\mathcal{L}_\theta)$
is called {\em generating set} for the given system.
Since the operations of symmetry transformation and taking the linear combination commute, the above definition is well-posed.

Generating sets of such type, called also $L$-bases, were firstly described in~\cite{Khamitova1982} (see also~\cite{Ibragimov1985}).

\medskip

\noindent{\bf 3. Generating sets of pairs ``system + conservation law''.}
Consider now the action of an equivalence transformation on a pair from~$P$.
If the first element of the resulting pair coincides with the first element of another pair from~$P$,
we can take linear combinations of their second elements.
All conservation laws that cannot be obtained in such way will be called inequivalent.
In such way we obtain a generating set of pairs ``system + conservation law''.

It becomes especially interesting if such generating set
depends on arbitrary elements (parameters) of the system.
In particular, it is possible that the characteristics of conservation laws do not depend on the parameters.
Then, the action of equivalence group on a pair from the generating set maps it to a pair from~$P$.
Taking the linear combinations of conservation laws for the pairs with coinciding first elements,
we get all possible conservation laws.

\begin{note}
A notion of generating sets of pairs ``system + conservation law'' with respect to all admissible transformations
is introduced in an analogous way.
\end{note}

Note that a problem of constructing such generating sets is different from the one of classification of conservation laws
with respect to equivalence group: under the problem of classification of conservation laws with respect to
equivalence group we understand classification of pairs ``equation+ conservation law'',
and a conservation law from the generating set can generate a subspace of the space of conservation laws
for equation with another values of arbitrary elements.

\begin{note}
It can be easy shown that all the above equivalences are indeed equivalence relations,
i.e., they have the usual reflexive, symmetric and transitive properties.
\end{note}

\section{Local conservation laws of diffusion--convection equations}\label{SectionOnConsLawsOfEqDKfgh}

We search (local) conservation laws of equations from class~\eqref{eqDKfh},
applying the modification of the direct method, which was proposed in~\cite{Popovych&Ivanova2004ConsLawsLanl}
and is based on using the notion of equivalence of conservation laws with respect to a transformation group
and classification up to the equivalence group of a class of differential equations.

Conservation laws were investigated for some subclasses of class~\eqref{eqDKfgh}.
In~particular, conservation laws of the linear heat equation and Burgers equation were constructed firstly in~\cite{Atherton&Homsy1975}.
V.A.~Dorodnitsyn and S.R.~Svirshchevskii~\cite{Dorodnitsyn&Svirshchevskii1983}
(see also~\cite[Chapter~10]{Ibragimov1994V1})
completely investigated the local conservation laws for one-dimensional reaction--diffusion equations $u_t=(A(u))_{xx}+C(u)$
having non-empty intersection with class~\eqref{eqDKfgh}.
In~\cite{Kara&Mahomed2002} the~first-order local conservation laws of equations~\eqref{eqDKfgh} are found.
Developing results obtained in~\cite{Bluman&Doran-Wu1995} for the case~$hB=0$, $f=1$,
in the recent papers~\cite{Ivanova2004,Popovych&Ivanova2004ConsLawsLanl}
we completely classified potential conservation laws (including arbitrary order local ones)
of equations~\eqref{eqDKfgh} with $f=g=h=1$ with respect to the corresponding equivalence group.
%We also constructed an exhaustive list of locally inequivalent potential systems corresponding to these equations.
Conservation laws of equations~\eqref{eqDKfgh} were considered in~\cite{Ivanova&Popovych&Sophocleous2004}.

In view of results of Section~\ref{SectionOnEquivOfConsLaws}
it is sufficient for exhaustive investigation if
we classify conservation laws of equations only from class~(\ref{eqDKfh}).

There are two independent variables~$t$ and~$x$ in equations under consideration, which play a role of
the time and space variables correspondingly.
Therefore, the general form of constructed conservation laws will be
\begin{equation}\label{EqGen2DimConsLaw}
D_tF(t,x,u_{(r)})+D_xG(t,x,u_{(r)})=0,
\end{equation}
where $D_t$ and $D_x$ are the operators of total differentiation with respect to $t$ and $x$.
The components $F$ and $G$ of the conserved vector~$(F,G)$ are called the {\em density} and the {\em flux}
of the conservation law.
Two conserved vectors $(F,G)$ and $(F',G')$ are equivalent if
there exist such functions~$\hat F$, $\hat G$ and~$H$ of~$t$, $x$ and derivatives of~$u$ that
$\hat F$ and $\hat G$ vanish for all solutions of~$\mathcal{L}$~and $F'=F+\hat F+D_xH$, $G'=G+\hat G-D_tH$.

At first we prove the lemma on order of local conservation laws for more general class
of second-order evolution equations, which covers class~\eqref{eqDKfgh}.

\begin{lemma}\label{LemmaOnOrderOfConsLawsOfDCEs}
Any local conservation law of any second-order $(1+1)$-dimensional quasi-linear evolutionary equation has the first order
and, moreover, there exists its conserved vector with the density depending at most on $t$, $x$, and $u$
and the flux depending at most on $t$, $x$, $u$ and~$u_x$.
\end{lemma}

\begin{proof}
Consider a conservation law of form~\eqref{EqGen2DimConsLaw} of the second-order $(1+1)$-dimensional
quasi-linear evolutionary equation
\begin{equation}\label{gso11dimee}
u_t=S(t,x,u,u_x)u_{xx}+R(t,x,u,u_x),
\end{equation}
where $S\ne0$.
In view of equation~\eqref{gso11dimee} and its differential consequences,
we can assume that $F$ and~$G$ depend only on $t$, $x$ and $u_k=\partial^k u/\partial x^k$, $k=\overline{0,r'},$
where $r'\le 2r$. Suppose that $r'>1$.
We~expand the total derivatives in (\ref{EqGen2DimConsLaw}) and take into account differential consequences of
form $u_{tj}=D_x^{j}(Su_{xx}+R)$, where $u_{tj}=\partial^{j+1} u/\partial t\partial x^k$, $j=\overline{0,r'}$.
As a result, we obtain the following condition
\begin{equation}\label{clcdeom}\textstyle
F_t+F_{u_j}D_x^{j}(Su_{xx}+R)+G_x+G_{u_j}u_{j+1}=0.
\end{equation}
Let us decompose~\eqref{clcdeom} with respect to the highest derivatives $u_j$.
Thus, the coefficients of $u_{r'+2}$ and $u_{r'+1}$ give the equations
$F_{u_{r'}}=0$, $G_{u_{r'}}+SF_{u_{r'-1}}=0$ that implies
\[
F=\hat F, \quad G=-S\hat F_{u_{r'-1}}u_{r'}+\hat G,
\]
where $\hat F$ and $\hat G$ are functions of $t$, $x$, $u$, $u_1$, \ldots, $u_{r'-1}$.
Then, after selecting the terms containing $u_{r'}^2$, we obtain that $-S\hat F_{u_{r'-1}u_{r'-1}}=0$.
It yields that $\hat F =\check F^1u_{r'-1}+\check F^0,$
where $\check F^1$ and $\check F^0$ depend only on $t$, $x$, $u$, $u_1$,~\ldots, $u_{r'-2}$.

Consider the conserved vector with the density~$\tilde F=F-D_xH$ and the flux~$\tilde G=G+D_tH$,
where $H=\int \check F^1du_{r'-2}$. This conserved vector is equivalent to the initial one, and
\[
\tilde F=\tilde F(t,x,u,u_1,\ldots,u_{r'-2}), \quad
\tilde G=\tilde G(t,x,u,u_1,\ldots,u_{r'-1}).
\]

Iterating the above procedure a necessary number of times, we result in an equivalent conserved vector
depending only on $t$, $x$, $u$ and $u_x$, i.e. we can assume at once that $r'\le 1$.
Then the coefficients of $u_{xxx}$ and $u_{xx}$ in~\eqref{clcdeom} lead to the equations
$F_{u_x}=0$, $G_{u_x}+SF_u=0$ that implies $F=F(t,x,u)$ and, moreover, $G=-F_u\int\!S\,du_x+G^1$,
where $G^1=G^1(t,x,u)$.
\end{proof}

\begin{note}
A similar statement is true for an arbitrary (1+1)-dimensional evolution equation~$\cal L$ of the even
order~$r=2\bar r$, $\bar r\in\mathbb{N}$. For example~\cite{Ibragimov1985}, for any conservation law of~$\cal L$
we can assume up to equivalence of conserved vectors
that $F$ and $G$ depend only on~$t$, $x$ and derivatives of~$u$ with respect to~$x$, and
the maximal order of derivatives in~$F$ is not greater than $\bar r$.

Lemma~\ref{LemmaOnOrderOfConsLawsOfDCEs} gives a stronger result for a more restricted class of equations.
In the above proof we specially use the most direct method to demonstrate its effectiveness in
quite general cases. This proof can be easily extended to other classes of (1+1)-dimensional evolution equations
of even orders and some systems connected with evolution equations~\cite{Popovych&Ivanova2004ConsLawsLanl}.
\end{note}

\setcounter{mcasenum}{0}

\begin{theorem}\label{TheoremOnClassificationCLsSmallGroup}
A complete list of $G^{\sim}_1$-inequivalent equations~\eqref{eqDKfh} having nontrivial
conservation laws is exhausted by the following ones

\vspace{1ex}

$\makebox[6mm][l]{\refstepcounter{mcasenum}\themcasenum\label{h1}.}
 h=1 \colon\quad (\,fu,\ -Au_x-\int\!\! B\,).$

%\vspace{1ex}

%$\makebox[6mm][l]{\refstepcounter{mcasenum}\themcasenum\label{hx-1}.} h=x^{-1} \colon\quad
%(\,xfu,\ -xAu_x+\int\!\!A-\int\!\!B\,).$

\vspace{1ex}

$\makebox[6mm][l]{\refstepcounter{mcasenum}\themcasenum\label{A1Bne0}.}
 A=1, \quad B_u\ne0,\quad f=-h(h^{-1})_{xx}\colon \quad \textstyle
(\,e^t(h^{-1})_{xx} u,\ e^t(h^{-1}u_x-(h^{-1})_xu+\int\!\! B)\,). $

\vspace{1ex}

$\makebox[6mm][l]{\refstepcounter{mcasenum}\themcasenum\label{B1fhx}.}
 B=\varepsilon A+1, \quad f=h_y\colon \quad
 (\,e^{t-\varepsilon\int\!\! h}h_yu,\ -e^{t}(Ae^{-\varepsilon\int\!\! h}u_y+hu)\,).
$

\vspace{1ex}

$\makebox[6mm][l]{\refstepcounter{mcasenum}\themcasenum\label{B1fhx+hx-1}.} B=\varepsilon A+1, \quad f=h_y+hy^{-1}\colon
\quad\textstyle
(\,e^{t-\varepsilon\int\!\! h}yfu ,\ -e^t(yAe^{-\varepsilon\int\!\! h}u_y+yhu-\int\!\! A)\,).$

\vspace{1ex}

$\makebox[6mm][l]{\refstepcounter{mcasenum}\themcasenum\label{B0}.}
B=\varepsilon A\colon \quad
(yfe^{-\varepsilon\!\int\!\! h} u,\ -yAe^{-\varepsilon\!\int\!\! h}u_y+\int\!\! A),
\ (fe^{-\varepsilon\!\int\!\! h}u,\ -Ae^{-\varepsilon\!\int\!\! h}u_y).$

\vspace{1ex}

$\makebox[6mm][l]{\refstepcounter{mcasenum}\themcasenum\label{B1fhint}.}\displaystyle
B=\varepsilon A+1, \quad f=-Z^{-1}h,
\quad h=Z^{-1/2}\exp\left(-\int\dfrac{a_{00}+a_{11}}{2Z}dy\right)\colon$

\vspace{1ex}

$\makebox[6mm][l]{}(\,(\sigma^{k1}y+\sigma^{k0})fe^{-\varepsilon\int\!\! h}u,\;
-(\sigma^{k1}y+\sigma^{k0})(Ae^{-\varepsilon\int\!\! h}u_y+hu)+\sigma^{k1}\int\!\! A\,).$

\vspace{1ex}

$\makebox[6mm][l]{\refstepcounter{mcasenum}\themcasenum\label{A1B0}.}
 A=1, \quad B=0\colon \quad (\,\alpha f u, \ -\alpha u_x+\alpha_x u\,).$

\vspace{1ex}

\noindent Here
$y$ is implicitly determined by the formula $x=\int\! e^{\varepsilon \int\!h(y)dy}dy$;
$\varepsilon,a_{ij}=\rm const$, $i,j=\overline{0,1}$;
$(\sigma^{k1},\sigma^{k0})=(\sigma^{k1}(t),\sigma^{k0}(t))$, $k=\overline{1,2}$, is a fundamental solution system
of the system of ODEs $\sigma^\nu_t=a_{\mu\nu}\sigma^\mu$;
$Z=a_{01}y^2+(a_{00}-a_{11})y-a_{10}$;
$\alpha=\alpha(t,x)$ is an arbitrary solution of the linear equation
$f\alpha_t+\alpha_{xx}=0$. Hereafter $\int\! A=\int\!A\,du$, $\int\! B=\int\!B\,du$.
In case~\ref{B0} $\varepsilon\in\{0,1\}\!\!\mod G^{\sim}_1$.

(Together with constraints on the parameter-functions $A$, $B$, $f$ and $h$
we also adduce conserved vectors of the basis elements of the corresponding space of conservation laws.)
\end{theorem}

In Theorem~\ref{TheoremOnClassificationCLsSmallGroup} we have classified conservation laws
with respect to the usual equivalence group $G^{\sim}_1$.
The obtained result can be formulated in an implicit form only,
and indeed case~\ref{B1fhint} is split into a number of inequivalent cases depending on values of~$a_{ij}$.
At the same time, using the
extended equivalence group~$\hat G^{\sim}_1$,
we can present the result of classification in a closed and simple form with a smaller number
of inequivalent equations having nontrivial conservation laws.

\begin{theorem}\label{TheoremOnClassificationCLsWideGroup}
A complete list of $\hat G^{\sim}_1$-inequivalent equations~\eqref{eqDKfh} having nontrivial
conservation laws is exhausted by the following ones
\setcounter{mcasenum}{0}

\vspace{1ex}

$\makebox[6mm][l]{\refstepcounter{mcasenum}\themcasenum\label{2.h1}.}
h=1 \colon\quad (\,fu,\ -Au_x-\int\!\! B\,),\ 1.$

\vspace{1ex}

$\makebox[6mm][l]{\refstepcounter{mcasenum}\themcasenum\label{2.A1Bune0}.}
A=1, \quad B_u\ne0,\quad f=-h(h^{-1})_{xx}\colon \quad \textstyle
(\,e^t(h^{-1})_{xx} u,\ e^t(h^{-1}u_x-(h^{-1})_xu+\int\!\! B)\,),\ -e^th^{-1}.$

\vspace{1ex}

$\makebox[6mm][l]{\refstepcounter{mcasenum}\themcasenum\label{2.B1fhx}.}
B=1, \quad f=h_x\colon \quad (\,e^tfu,\ -e^t(Au_x+hu)\,),\ e^t.$

\vspace{1ex}

$\makebox[6mm][l]{\refstepcounter{mcasenum}\themcasenum\label{2.B1fhx+hx-1}.}
B=1, \quad f=h_x+hx^{-1}\colon \quad (\,e^txfu ,\ -e^t(xAu_x+xhu-\int\!\! A)\,),\ e^tx .$

\vspace{1ex}

$\makebox[6mm][l]{\refstepcounter{mcasenum}\themcasenum\label{2.B0}{\rm a}.}
B=0\colon\quad (\,fu,\ -Au_x\,),\ 1;\ (\,xfu,\ -xAu_x+\int\!\! A\,),\ x.$

\vspace{1ex}

$\makebox[6mm][l]{\themcasenum{\rm b.}}
B=1, \quad f=1,\quad h=1\colon \quad (\,u ,\ -Au_x-u\,),\ 1;\ \textstyle
(\,(x+t)u ,\ -(x+t)(Au_x+u)+\int\!\! A\,),\ x+t.$

\vspace{1ex}

$\makebox[6mm][l]{\themcasenum{\rm c}.} B=1,\quad f=e^x,\quad h=e^x\colon \quad (\,e^{x+t}u ,\ -e^t(Au_x+e^xu)\,),\ e^t;$
\nopagebreak\\[1ex]\hspace*{\parindent}%
$\makebox[6mm][l]{} (\,e^{x+t}(x+t)u ,\ -e^t(x+t)(Au_x+e^xu)+e^t\int\!\! A\,),\ e^t(x+t).$

\vspace{1ex}

$\makebox[6mm][l]{\themcasenum{\rm d}.} B=1, \quad f=x^{\mu-1},\quad h=x^\mu\colon\quad
(\,e^{\mu t}x^{\mu-1}u ,\ -e^{\mu t}(Au_x+x^\mu u)\,),\ e^{\mu t};$
\nopagebreak\\[1ex]\hspace*{\parindent}%
$\makebox[6mm][l]{} (\,e^{(\mu+1)t}x^\mu u ,\ e^{(\mu+1)t}(-xAu_x-x^{\mu+1} u+\int\!\! A)\,),\ e^{(\mu+1)t}x.$

\vspace{1ex}

$\makebox[6mm][l]{\refstepcounter{mcasenum}\themcasenum\label{2.B1fexphexp}.} B=1,
\quad f=e^{-\mu/x}x^{-3},\quad h=e^{-\mu/x}x^{-1},\quad \mu\in\{0,1\}\colon$
\nopagebreak\\[1ex]\hspace*{\parindent}%
$\makebox[6mm][l]{} (\,e^{\mu t}xfu ,\ -e^{\mu t}x(Au_x+hu)+e^{\mu t}\int\!\! A\,),\ e^{\mu t}x;$
\nopagebreak\\[1ex]\hspace*{\parindent}%
$\makebox[6mm][l]{} (\,e^{\mu t}(tx-1)fu ,\ -e^{\mu t}(tx-1)(Au_x+hu)+te^{\mu t}\int\!\! A\,),\ e^{\mu t}(tx-1).$

\vspace{1ex}

$\makebox[6mm][l]{\refstepcounter{mcasenum}\themcasenum\label{2.B1fx1x1hx1x1}.} B=1,
\ f=|x-1|^{\mu-3/2}|x+1|^{-\mu-3/2},\ h=|x-1|^{\mu-1/2}|x+1|^{-\mu-1/2}\colon$
\nopagebreak\\[1ex]\hspace*{\parindent}%
$\makebox[6mm][l]{} (\,e^{(2\mu+1)t}(x-1)fu ,\ -e^{(2\mu+1)t}(x-1)(Au_x+hu)+e^{(2\mu+1)t}\int\!\! A\,),\
e^{(2\mu+1)t}(x-1);$
\nopagebreak\\[1ex]\hspace*{\parindent}%
$\makebox[6mm][l]{} (\,e^{(2\mu-1)t}(x+1)fu ,\ -e^{(2\mu-1)t}(x+1)(Au_x+hu)+e^{(2\mu-1)t}\int\!\! A\,),\
e^{(2\mu-1)t}(x+1).$

\vspace{1ex}

$\makebox[6mm][l]{\refstepcounter{mcasenum}\themcasenum\label{2.B1farctanharctan}.}
B=1, \quad f=e^{\mu\arctan x}(x^2+1)^{-3/2},
\quad h=e^{\mu\arctan x}(x^2+1)^{-1/2}\colon$
\nopagebreak\\[1ex]\hspace*{\parindent}%
$\makebox[6mm][l]{} (\,e^{\mu t}(x\cos t+\sin t)fu ,\ -e^{\mu t}(x\cos t+\sin t)(Au_x+hu)+e^{\mu t}\cos t\int\!\! A\,),\
e^{\mu t}(x\cos t+\sin t);$
\nopagebreak\\[1ex]\hspace*{\parindent}%
$\makebox[6mm][l]{} (\,e^{\mu t}(x\sin t-\cos t)fu, \ -e^{\mu t}(x\sin t-\cos t)(Au_x+hu)+e^{\mu t}\sin t\int\!\! A\,),\
e^{\mu t}(x\sin t-\cos t).$

\vspace{1ex}

$\makebox[6mm][l]{\refstepcounter{mcasenum}\themcasenum\label{2.A1B0}.}
A=1, \quad B=0\colon \quad (\,\alpha f u, \ -\alpha u_x+\alpha_x u\,),\ \alpha.$

\vspace{1ex}

\noindent
Here $\mu={\rm const}$,
$\alpha=\alpha(t,x)$ is an arbitrary solution of the linear equation \mbox{$f\alpha_t+\alpha_{xx}=0$}.
(Together with constraints on the parameter-functions $A$, $B$, $f$ and $h$
we also adduce conserved vectors and characteristics
of the basis elements of the corresponding space of conservation laws.)
\end{theorem}

\begin{proof}
In view of lemma~\ref{LemmaOnOrderOfConsLawsOfDCEs}, we can assume at once that $F=F(t,x,u)$ and $G=G(t,x,u,u_x)$.
Let us substitute the expression for~$u_t$ deduced from~\eqref{eqDKfh}
into~(\ref{EqGen2DimConsLaw}) and decompose the obtained equation with respect to $u_{xx}$. The coefficient of~$u_{xx}$
gives the equation $AF_u+fG_{u_x}=0$, therefore $G=-Af^{-1}F_{u}u_x+G^1(t,x,u)$.
Taking into account the latter expression for $G$ and splitting the rest of equation~(\ref{EqGen2DimConsLaw})
with respect to the powers of $u_x$, we obtain the system of PDEs for the functions
$F$ and $G^1$ of the form
\begin{equation}\label{splitconslaw}
F_{uu}=0, \quad \frac hf BF_u-A\Bigl(\frac{F_{u}}f\Bigr)_x+G^1_u=0, \quad F_t+G^1_x=0.
\end{equation}
Solving first two equations of~(\ref{splitconslaw}) yields
\[%\label{efFaG}
\textstyle
F=F^1(t,x)u+F^0(t,x), \quad G^1=\Bigl(\dfrac{F^1}f\Bigr)_x\int\! A-\dfrac hfF^1\int\! B+G^0(t,x).
\]
In further consideration the major role is played by a differential consequence of system~(\ref{splitconslaw})
that can be written as
\begin{equation}\label{fghABClassifyingConditionForConsLaws}
A\Bigl(\dfrac{F^1}f\Bigr)_{xx}-B\Bigl(\dfrac hfF^1\Bigr)_x+F^1_t=0.
\end{equation}
Indeed, it is the unique classifying condition for this problem.
In all classification cases we obtain the equation $F^0_t+G^0_x=0$.
Therefore, up to conserved vectors equivalence we can assume $F^0=G^0=0$,
and additionally $F^1\ne0$ for conservation laws to be non-trivial.
Equation~\eqref{fghABClassifyingConditionForConsLaws} implies that
there exist no non-trivial conservation laws in the general case.
Let us classify the special values of the parameter-functions
for which equation~\eqref{eqDKfh} possesses non-trivial conservation laws.
There exist four different possibilities for values of $A$ and $B$.

\vspace{0.8ex}

1. $\dim \langle A,B,1\rangle=3$. It follows from~\eqref{fghABClassifyingConditionForConsLaws} that
$F^1_t=(F^1/f)_{xx}=(hF^1/f)_x=0$ and therefore $F^1=C_1f/h$, $(1/h)_{xx}=0$, i.e., obviously
$h\in\{1,x^{-1}\}\!\!\mod G^{\Equiv}$. Moreover, $h=1\sim \tilde h=x^{-1}\!\!\mod G^{\Equiv}$
(the corresponding transformation is $\tilde x =\ln|x|$ and $\tilde f =x^2f$,
the other variables and parameter-functions are not changed).
As a result, we obtain case~\ref{2.h1}.

\vspace{0.8ex}

2. $A\in\langle 1\rangle$, $B\not\in\langle 1\rangle$.
Then $A=1\!\!\mod G^{\Equiv}$ and $(hF^1/f)_x=0$, $F^1_t+(F^1/f)_{xx}=0$,
i.e., $F^1=\alpha(t)f/h$, where $\alpha_t/\alpha=\varkappa=\const.$
$\varkappa\ne0$ (otherwise we have case~\ref{2.h1}) and so
$\varkappa=1$, $f=-h(h^{-1})_{xx}\!\!\mod G^{\Equiv}$
(case~\ref{2.A1Bune0}).

\vspace{0.8ex}

3. $A\not\in\langle 1\rangle$, $B\in\langle A, 1\rangle$.
Then $B\in\langle 1\rangle\!\!\mod G^{\Equiv}$ and $(F^1/f)_{xx}=0$, $F^1_t=B(hF^1/f)_x$,
i.e. $F^1=(\alpha^1(t)x+\alpha^0(t))f$ and $\alpha^1_txf+\alpha^0_tf=B(\alpha^1(xh)_x+\alpha^0h_x)$.
For $B=0$ we obtain case~\ref{2.B0}{\rm a} at once.
Suppose $B\ne0$. Then $B=1\!\!\mod G^{\Equiv}$
and the dimension $m=\dim\langle f,xf,h_x,(xh)_x\rangle$ can have only the values 2 and 3.

If $m=3$ then there exist constants~$a_{\mu\nu}$, $b_\mu$, $\mu,\nu=0,1$,
and a function $\theta=\theta(x)$ such that $(b_0,b_1)\ne(0,0)$, $\dim\langle f,xf,\theta\rangle=3$ and
$h_x=a_{00}f+a_{01}xf+b_0\theta$, $(xh)_x=a_{10}f+a_{11}xf+b_1\theta$.
Therefore, $\alpha^\nu_t=a_{\mu\nu}\alpha^\mu$, $b_\mu\alpha^\mu=0$, i.e. $\alpha^\mu=C_\mu e^{\sigma t}$,
where $C_\mu,\sigma=\const$ and $\sigma\ne0$ (otherwise, this case is reduced to a subcase of~\ref{2.h1}),
hence $\sigma=1\!\!\mod G^{\Equiv}$. Depending on values (either vanishing or non-vanishing) of~$C_1$
we obtain cases~\ref{2.B1fhx} and~\ref{2.B1fhx+hx-1} correspondingly.

If $m=2$ then $h_x=a_{00}f+a_{01}xf$, $(xh)_x=a_{10}f+a_{11}xf$ for some constants~$a_{\mu\nu}$, $\mu,\nu=0,1$.
Therefore, $\alpha^\nu_t=a_{\mu\nu}\alpha^\mu$, $f=-h/Z$, $h_x/h=-(a_{01}x+a_{00})/Z$, where
$Z=a_{01}x^2+(a_{00}-a_{11})x-a_{10}$, i.e.,
$h=Z^{-1/2}\exp(-\frac12 a_{\mu\mu}\int Z^{-1}dx)$.
As a results, we obtain two conservation laws with the conserved vectors
\[\textstyle
(\alpha^{i1}(t)x+\alpha^{i0}(t))fu,\ -(\alpha^{i1}(t)x+\alpha^{i0}(t))(Au_x+hu)+\alpha^{i1}(t)\int\!\! A),
\]
where $(\alpha^{i1},\alpha^{i0})$, $i=1,2$ form a fundamental set of solutions
of the system $\alpha^\nu_t=a_{\mu\nu}\alpha^\mu$.
Separate consideration of possible inequivalent values of the constants~$a_{\mu\nu}$ leads to
cases~\ref{2.B0}b--\ref{2.B0}d and~\ref{2.B1fexphexp}--\ref{2.B1farctanharctan}.

\vspace{0.8ex}

4. $A,B\in\langle 1\rangle$. Therefore, $A=1$, $B=0\!\!\mod\hat G^{\Equiv}$ and $F^1_t+(F^1/f)_{xx}=0$
(case~\ref{2.A1B0}).
\end{proof}

Analysis of the classification given
in theorems~\ref{TheoremOnClassificationCLsSmallGroup} and~\ref{TheoremOnClassificationCLsWideGroup}
results in the following conclusions.

There are no conservation laws for equations from class~\eqref{eqDKfh} in the general case.
Imposing restrictions on values of arbitrary parameters leads only to subclasses of~\eqref{eqDKfh} with
one-, two- and infinite-dimensional spaces of (local) conservation laws.
An equation from class~\eqref{eqDKfh} has infinite number of linearly independent conservation laws iff
it is linear.
In this case the space of conservation laws is parameterized with an arbitrary solution of the corresponding
backward equation (which is obtained from the initial one by the time reflection).

There are four subclasses of equations~\eqref{eqDKfh} with one linearly independent conservation law
(cases~\ref{2.h1}--\ref{2.B1fhx+hx-1}).
Case~\ref{2.h1} is distinguished from the other ones due to the single constraint~$h=\const$ and, therefore,
arbitrariness of nonlinearities~$A$ and~$B$ and the parameter-function~$f$.
The distinguishing features of case~\ref{2.A1Bune0} are disappearance of nonlinearity with respect to~$A$ due to
the constraint~$A=\const$, arbitrariness of~$B$ and a second-order differential constraint between~$f$ and~$h$.
All the latter cases are obtained with the constraint~$B\in\langle A, 1\rangle$ and first-order
differential constraints of~$f$ and~$h$
since only cases~\ref{2.B1fhx} and~\ref{2.B1fhx+hx-1} with the constraint~$B\in\langle A, 1\rangle$
admit extensions of space of conservation laws.

Cases \ref{2.B0}b--\ref{2.B0}d of Theorem~\ref{TheoremOnClassificationCLsWideGroup} can be reduced to case \ref{2.B0}a
by means of additional equivalence transformations which belong neither to~$G^{\sim}$ nor even to~$\hat G^{\sim}$
and are pure form-preserving point transformations%
~\cite{Kingston&Sophocleous1991,Kingston&Sophocleous1998,Kingston&Sophocleous2001} for class~\eqref{eqDKfh}:
\begin{gather}\nonumber
\ref{2.B0}{\rm b}\to\ref{2.B0}{\rm a}_{f=1}\colon\quad \tilde t=t,\quad
  \tilde x=x+t,\quad \tilde u=u;\\ \nonumber
\ref{2.B0}{\rm c}\to\ref{2.B0}{\rm a}_{f=e^x}\colon\quad \tilde t=e^t,\quad
  \tilde x=x+t,\quad \tilde u=u;\\ \nonumber
\ref{2.B0}{\rm d} (\mu+1\ne0)\to\ref{2.B0}{\rm a}_{f=x^{\mu-1}}\colon\quad
  \tilde t=(\mu+1)^{-1}(e^{(\mu+1)t}-1),\quad \tilde x=xe^t,\quad \tilde u=u;\\
\ref{2.B0}{\rm d} (\mu+1=0)\to\ref{2.B0}{\rm a}_{f=x^{-2}}\colon\quad \tilde t=t,\quad
  \tilde x=xe^t,\quad \tilde u=u.\label{TransAdEquivTransOfCLs}
\end{gather}
%\begin{gather*}
%\ref{2.B0}{\rm d}=scale(\mu+1)4 \cap scale(\mu)3 \qquad  \ref{2.B0}{\rm d}\to 5 b \ or\ 5c\\
%6=?\\
%7=4^{-1}\cap4^{+1},\qquad 7\to 6\\
%7=4^{-i}\cap4^{+i},%
%\end{gather*}

\section{Contractions of conservation laws}\label{SectionOnContractionsOfCLsOfDKfgh}

Similarly to contractions of symmetries or equations~\cite{Ivanova&Popovych&Sophocleous2006Part2} one can consider the notion
of contractions of conservation laws.

Consider the class~$\{\mathcal{L}(\varepsilon)\}$ of systems~$\mathcal{L}(\varepsilon)$:
$L(x,u_{(\prho)},\varepsilon)=0$
of $l$~differential equations for $m$~unknown functions $u=(u^1,\ldots,u^m)$
of $n$~independent variables $x=(x_1,\ldots,x_n)$,
which are parameterized with the parameter~$\varepsilon.$
Here $u_{(\prho)}$ denotes the set of all the derivatives of~$u$ with respect to $x$
of order not greater than~$\prho$, including $u$ as the derivatives of the zero order.
$L=(L^1,\ldots,L^l)$ is a tuple of $l$ fixed functions depending on $x,$ $u_{(\prho)}$ and $\varepsilon$.
%For simplicity we assume $\varepsilon$ as a single numeric (real or complex) parameter.
%(Extension to more general case with respect to $\varepsilon$ is obvious.)
%
Let $\Lambda(\varepsilon)=(\lambda^1(x,u_{(k)},\varepsilon),\ldots,\lambda^l(x,u_{(k)},\varepsilon))$ are characteristics of conservation laws of
the systems~$\mathcal{L}(\varepsilon)$, i.e., $\Lambda(\varepsilon)\in\Ch(\mathcal{L}(\varepsilon))$.
Suppose also that
$\Lambda(\varepsilon)\to\bar\Lambda$, $\mathcal{L}(\varepsilon)\to\bar{\mathcal{L}}$, $\varepsilon\to0$ in $C^k(J^\prho)$.

Consider the action of Euler operator~$\Eop$ on characteristic forms of the conservation law of the systems~$\mathcal{L}(\varepsilon)$:
$\Eop\lambda^\mu(\varepsilon)L^\mu(\varepsilon)=0$ $\forall\varepsilon$.
Since $\lambda^\mu(\varepsilon)L^\mu(\varepsilon)\to\bar\lambda^\mu\bar L^\mu$, $\varepsilon\to0$ in $C^k(J^\prho)$
then
$0=\Eop(\lambda^\mu(\varepsilon)L^\mu(\varepsilon))\to\Eop(\bar\Lambda^\mu\bar L^\mu)=0$,
$\varepsilon\to0$. Therefore $\Eop(\bar\lambda^\mu\bar L^\mu)=0$ or $\bar\Lambda\in\Ch(\bar{\mathcal{L}})$
and $\bar\lambda^\mu\bar L^\mu=0$ is a characteristic form of the conservation law of the system~$\bar{\mathcal{L}}$.

By analogy with terminology accepted for Lie algebras,
we will call such limits as {\it contractions of characteristics} and {\it contractions of conservation laws}.

\begin{example}
Consider the equation
\[
x^{\mu-1}u_t=(Au_x)_x+x^\mu u_x
\]
having two linearly independent conservation laws with the characteristics $\lambda_1^\mu=e^{\mu t}$ and $\lambda_2^\mu=e^{(\mu+1)t}x$
(Case~\ref{2.B0}.d of Theorem~\ref{TheoremOnClassificationCLsWideGroup}).
Under the contraction $x=1+\tilde x/\mu$, $t=\tilde t/\mu$,  $\mu\to+\infty$ it goes to the equation
%$A=\tilde A/\mu$,
\[
e^{\tilde x}u_{\tilde t}=(Au_{\tilde x})_{\tilde x}+e^{\tilde x}u_{\tilde x}
\]
which possesses linearly independent conservation laws
with the characteristics $\bar\lambda_1=e^t$ and $\bar\lambda_2=e^{\tilde t}(\tilde x+\tilde t)$
(Case~\ref{2.B0}.c of Theorem~\ref{TheoremOnClassificationCLsWideGroup}).

Under the same limit process the characteristics $\lambda_1^\mu$ and $\lambda_2^\mu$ are transformed to characteristics of the target equation.
Name
$\lambda_1^\mu=e^{\mu t}=e^{\tilde t}\to e^{\tilde t}=\bar\lambda_1$ and
$\mu(\lambda_1^\mu-\lambda_2^\mu)=e^{\tilde t}((e^{\tilde t/\mu}-1)\mu+\tilde xe^{\tilde t/\mu})\to\bar\lambda_2=e^{\tilde t}(\tilde x+\tilde t)$.
\end{example}

\begin{example}
\[
\left|\frac{x-1}{x+1}\right|^\mu|x^2-1|^{-3/2}u_t=(Au_x)_x+\varepsilon\left|\frac{x-1}{x+1}\right|^\mu|x^2-1|^{-1/2}u_x
\]
Under contraction $x=2\mu\tilde x/\mu'$, $t=\mu'\tilde t/(2\mu)$,  $\mu\to+\infty$ it maps to
\[
\tilde x^{-3}e^{-\mu'/x}u_{\tilde t}=(Au_{\tilde x})_{\tilde x}+\varepsilon\tilde x^{-1}e^{-\mu'/x}u_{\tilde x}
\]
The characteristics $\lambda_1^\mu=e^{(2\mu+1)t}(x-1)$ and $\lambda_2^\mu=e^{(2\mu-1)t}(x+1)$ of the conservation laws of the origin equation
give rise the characteristics for the target equation:
\[
\lambda_1^\mu=e^{(2\mu+1)\tilde t}(\tilde x-1)=e^{\mu'\tilde t}\frac{e^{\mu'\tilde t/(2\mu)}}{\mu'/(2\mu)}\left(\tilde x-\frac{\mu'}{2\mu}\right)
\to e^{\mu'\tilde t}x=\bar\lambda_1
\]
and
\[
\frac12(\lambda_1-\lambda_2)=
e^{\mu'\tilde t}\left(\tilde x\frac{e^{\mu'\tilde t/\mu}-1}{\mu'/\mu}-\frac12(e^{\mu'\tilde t/\mu}+1)\right)e^{\mu'\tilde t/(2\mu)}
\to e^{\mu'\tilde t}(\tilde t\tilde x-1)=\bar\lambda_2.
\]
\end{example}

The problem of finding all contractions of conservation laws of class~\eqref{eqDKfh} is
closely connected to the problem of construction of contractions of equations~\eqref{eqDKfh},
and therefore, still remains open.

\section{Generating sets of conservation laws\\ of nonlinear diffusion--convection equations}\label{SectionOnGenSetsOfCLsDCEs}

We emphasize that the conservation laws adduced in Theorem~\ref{TheoremOnClassificationCLsWideGroup}
are $G^{\sim}$-inequivalent and linearly independent.
At the same time, acting by the Lie symmetry transformations of equations having multidimensional spaces of conservation laws
or by the equivalence transformations of the whole class of equations,
one can find generating sets of conservation laws with respect to the symmetry groups.

\begin{note}
Since the cases~\ref{2.B0}b--\ref{2.B0}d of Theorem~\ref{TheoremOnClassificationCLsWideGroup} can be reduced to case~\ref{2.B0}a
by means of point transformations~\eqref{TransAdEquivTransOfCLs}, below we exclude them from the consideration
and assume them to be equivalent to case~\ref{2.B0}a.
\end{note}

\begin{note}
Below we investigate potential systems and potential conservation laws of nonlinear equations only.
The detailed analysis of linear case can be found in~\cite{Popovych&Kunzinger&Ivanova2007}.
\end{note}

\begin{example}
It is obvious that in general, for case~5a there exists no symmetry transformation changing the conservation laws
of the fixed equation.
However, for subcase $f=1$ of case~5a the space translation $x\to x+1$ maps the second basis conserved vector to
$(\,(x+1)u,\ -(x+1)Au_x+\int\!\! A\,)$. Subtracting it from the second basis vector, we obtain $(-u,-Au_x)$, that is the first basis vector
multiplied by $-1$. Therefore, the generating with respect to symmetry group set of conservation laws of the subclass
$f=1$, $B=0$ of equations~\eqref{eqDKfh} consists of the $D_t(xu)+D_x(\int\!\! A-xAu_x)=0$.
\end{example}

Investigation of generating sets of conservation laws can be performed both in terms of conserved vectors
(as in above example) or in terms of characteristics. Below for short of presentation we use the characteristics terminology.

\begin{example}
Consider case~\ref{2.B1fexphexp} of equations possessing two-dimensional space of conservation laws
with the basis conserved vectors
having characteristics $e^{\mu t}x$ and $e^{\mu t}(tx-1)$.
Symmetry transformation of time translation $t\to t+1$ maps the second characteristic
(modulo multiplying by a constant) to $e^{\mu t}(tx-x-1)$.
Subtracting the result from the characteristic $e^{\mu t}(tx-1)$ we get exactly the first basis characteristic $e^{\mu t}x$.
Therefore, the generating set of conservation laws of the subclass
$b=1$, $f=e^{-\mu/x}x^{-3}$, $g=1$, $h=e^{-\mu/x}x^{-1}$ of equations~\eqref{eqDKfh} consists of the conservation law
with characteristic $e^{\mu t}(tx-1)$.
\end{example}

Similarly one can prove that conservation laws in case~8 are generated by one with the characteristic $e^{\mu t}(x\cos t+\sin t)$
and in case~7 by the conservation law with the characteristic $e^{2\mu t}(x\cosh t-\sinh t)$.
Note that the basis conservation laws in this case have the characteristics~$e^{(2\mu-1)t}(x+1)$
and~$e^{(2\mu+1)t}(x-1)fu$. None of these two characteristics generates the whole space of conservation laws.

\medskip

Let us consider now the problem of construction of generating set of conservation laws
with respect to the extended equivalence group~$\hat G^{\sim}$.

\begin{example}
Consider subcase~\ref{2.B0}a ($B=0$) having two-dimensional space of conservation laws
spanned by ones with characteristics $1$ and $x$.
Subtracting the second characteristic from the result of application of equivalence transformation $\tilde x= x+1$ to it,
we get the first basis conservation law for equation $\tilde fu_t=(Au_{\tilde x})_{\tilde x}$ with~$\tilde f(\tilde x)= f(\tilde x-1)$.
Therefore, we can assume that conservation laws of subclass~$B=0$
are generated by~$D_t(xfu)-D_x(xAu_x-\int\!\! A)=0$ and group~$\hat G^{\sim}$ of equivalence transformations.
\end{example}

\begin{example}
Let us investigate in more detail case 7 of Theorem~\ref{TheoremOnClassificationCLsWideGroup}.
It is shown above that for the fixed value of~$\mu$
the set of local conservation laws is generated (with respect to the symmetry group of the equation)
by the conservation law with characteristic~$e^{2\mu t}(x\cosh t-\sinh t)$.
Alternatively, we can consider an action of equivalence transformations to the conservation laws.
Thus, discrete equivalence transformation $\tilde t=-t$, $\tilde x=-x$, $\tilde \mu=-\mu$
applied to the first basis characteristic~$e^{(2\mu+1)t}(x-1)$
generates the second characteristic~$e^{(2\tilde \mu-1)\tilde t}$
of equation with~$\tilde f=-f(-x,-\mu)$ $\tilde h=-h(-x,-\mu)$, $\tilde A=-A$ and~$\tilde\mu=-\mu$.
\end{example}

Similarly it can be shown that there exists no other nontrivial action of transformations from~$\hat G^{\sim}$ to conservation laws
of equations~\eqref{eqDKfh}. Therefore, the following theorem holds.

\begin{theorem}\label{TheoremGenSetCLEquivGr}
The generating with respect to~$\hat G^{\sim}$ set of nonlinear equations~\eqref{eqDKfh}
and corresponding conservation laws consists of
\setcounter{mcasenum}{0}

\vspace{1ex}

$\makebox[6mm][l]{\refstepcounter{mcasenum}\themcasenum.}
h=1 \colon\quad (\,fu,\ -Au_x-\int\!\! B\,)$;

\vspace{1ex}

$\makebox[6mm][l]{\refstepcounter{mcasenum}\themcasenum.}
A=1, \quad B_u\ne0,\quad f=-h(h^{-1})_{xx}\colon \quad \textstyle
(\,e^t(h^{-1})_{xx} u,\ e^t(h^{-1}u_x-(h^{-1})_xu+\int\!\! B)\,)$;

\vspace{1ex}

$\makebox[6mm][l]{\refstepcounter{mcasenum}\themcasenum.}
B=1, \quad f=h_x\colon \quad (\,e^tfu,\ -e^t(Au_x+hu)\,)$;

\vspace{1ex}

$\makebox[6mm][l]{\refstepcounter{mcasenum}\themcasenum.}
B=1, \quad f=h_x+hx^{-1}\colon \quad (\,e^txfu ,\ -e^t(xAu_x+xhu-\int\!\! A)\,)$;

\vspace{1ex}

$\makebox[6mm][l]{\refstepcounter{mcasenum}\themcasenum.}
B=0\colon\quad (\,xfu,\ -xAu_x+\int\!\! A\,)$;

\vspace{1ex}

$\makebox[6mm][l]{\refstepcounter{mcasenum}\themcasenum.} B=1,
\quad f=e^{-\mu/x}x^{-3},\quad h=e^{-\mu/x}x^{-1},\quad \mu\in\{0,1\}\colon$
\nopagebreak\\[1ex]\hspace*{\parindent}%
$\makebox[6mm][l]{} (\,e^{\mu t}(tx-1)fu ,\ -e^{\mu t}(tx-1)(Au_x+hu)+te^{\mu t}\int\!\! A\,)$;

\vspace{1ex}

$\makebox[6mm][l]{\refstepcounter{mcasenum}\themcasenum.} B=1,
\ f=|x-1|^{\mu-3/2}|x+1|^{-\mu-3/2},\ h=|x-1|^{\mu-1/2}|x+1|^{-\mu-1/2}\colon$
\nopagebreak\\[1ex]\hspace*{\parindent}%
$\makebox[6mm][l]{} (\,e^{(2\mu+1)t}(x-1)fu ,\ -e^{(2\mu+1)t}(x-1)(Au_x+hu)+e^{(2\mu+1)t}\int\!\! A\,)$;

\vspace{1ex}

$\makebox[6mm][l]{\refstepcounter{mcasenum}\themcasenum.}
B=1, \quad f=e^{\mu\arctan x}(x^2+1)^{-3/2},
\quad h=e^{\mu\arctan x}(x^2+1)^{-1/2}\colon$
\nopagebreak\\[1ex]\hspace*{\parindent}%
$\makebox[6mm][l]{} (\,e^{\mu t}(x\cos t+\sin t)fu ,\ -e^{\mu t}(x\cos t+\sin t)(Au_x+hu)+e^{\mu t}\cos t\int\!\! A\,).$

%\vspace{1ex}

\end{theorem}

\section{Potential systems and potential conservation laws}\label{SectionOnPotSymPotCLsTheory}

If the local conservation laws of a system~$\mathcal{L}$ of differential equations are known,
we may apply Lemma~\ref{lemma.null.divergence} to the constructed conservation laws
on the set of solutions of~\mbox{$\mathcal{L}={\mathcal L}^0$}.
In such way we introduce potentials as additional dependent variables. Then we attach the equations connecting the potentials with
components of corresponding conserved vectors to~${\mathcal L}^0$.
(If~\mbox{$n>2$} the attached equations of such kind form an underdetermined system with respect to the potentials.
Therefore, we can also attach gauge conditions on the potentials to~${\mathcal L}^0$.)

We have to use linearly independent conservation laws since otherwise the introduced potentials will be
{\em dependent} in the following sense: there exists a linear combination of the potential tuples,
which is, for some $r'\in{\mathbb N}$, a tuple of functions of $x$ and $u_{(r')}$ only.

Then we exclude the unnecessary equations (i.e., the equations that are dependent on
equations of~${\mathcal L}^0$ and attached equations simultaneously)
from the extended (potential) system~${\mathcal L}^1$
which will be called a {\em potential system of the first level}.
Any conservation law of~${\mathcal L}^0$ is a one of~${\mathcal L}^1$.
We iterate the above procedure for~${\mathcal L}^1$ to find its conservation laws
which are linearly independent with ones from the previous iteration
and will be called {\em potential conservation laws of the first level}.

We make iterations until it is possible
(i.e., the iteration procedure has to be stopped if all the conservation laws of
a {\em potential system~${\mathcal L}^{k+1}$ of the $(k+1)$-th level} are linearly dependent
with the ones of~${\mathcal L}^k$) or construct infinite chains of conservation laws by means of induction.
A such way may yield {\em purely potential} conservation laws of the initial system~$\mathcal L$,
which are linearly independent with local conservation laws and depend explicitly on potential variables.

Any conservation law from the previous step of iteration procedure will be a conservation law for the next step
and vice versa, conservation laws which are obtained on the next step
and depend only on variables of the previous step are linearly dependent with
conservation laws from the previous step.
It is also obvious that the conservation laws used for construction of a potential system of the next level are
trivial on the manifold of this system.

Since gauge conditions on potentials can be chosen in many different ways,
exhaustive realization of the above iteration procedure is improbable in case $n>2$.

The case of two independent variables is singular with respect to possible (constant) indeterminacy
after introduction of potentials and high effectiveness of application of potential symmetries.
That is why we consider some notions connected with conservation laws in this case separately.
We denote independent variables as $t$ (the time variable) and $x$ (the space one).
Any local conservation law has the form
\[
D_tF(t,x,u_{(r)})+D_xG(t,x,u_{(r)})=0.
\]
It allows us to introduce the new dependent (potential) variable~$v$
by means of the equations
\begin{equation}\label{potsys1}
v_x=F,\qquad v_t=-G.
\end{equation}

In the case of single equation~$\mathcal{L}$, equations of form~\eqref{potsys1} combine into
the complete potential system since~$\mathcal{L}$ is a differential consequence of~\eqref{potsys1}.
As a rule, systems of such kind admit a number of nontrivial symmetries and so they are of a great interest.

Lemma~\ref{PropositionOnInducedMapping} and equation~\eqref{potsys1} imply the following statement.

\begin{proposition}\label{2DConsLawEquivRelation}\cite{Popovych&Ivanova2004ConsLawsLanl}
Any point transformation connecting two systems~$\mathcal{L}$ and~$\tilde{\mathcal L}$
of PDEs with two independent variables generates a one-to-one mapping between the sets of potential systems,
which correspond to~$\mathcal{L}$ and~$\tilde{\mathcal L}$. Generation is made via trivial prolongation
on the space of introduced potential variables, i.e. we can assume that the potentials are not transformed.
\end{proposition}

\begin{corollary}
The Lie symmetry group of a system~$\mathcal{L}$ of differential equations generates an equivalence group
on the set of potential systems corresponding to~$\mathcal{L}$.
\end{corollary}

\begin{corollary}
Let $\widehat{\mathcal{L}}|_S$ be the set of all potential systems constructed
for systems from the class~$\mathcal{L}|_S$ with their conservation laws.
Action of transformations from~$G^{\Equiv}(L,S)$ together with the equivalence relation of potentials
naturally generates an equivalence relation on~$\widehat{\mathcal{L}}|_S$.
\end{corollary}

\begin{note}
Proposition~\ref{2DConsLawEquivRelation} and its corollaries imply that the equivalence group for a class of
systems or the symmetry group for single system can be prolonged to potential variables for any step of
the direct iteration procedure. It is natural the prolonged equivalence groups are used to classify
possible conservation laws and potential systems in each iteration.
Additional equivalences which exist in some subclasses of the class or arise
after introducing potential variables can be used for further analysis of connections between conservation laws.
\end{note}

\section{Potential systems of diffusion--convection equations}\label{SectionPotSysOfDifConvEq}

At first we consider potential systems obtained with introducing potentials associated to a single conservation law.
According to the terminology proposed in~\cite{Popovych&Ivanova2004ConsLawsLanl} we use the attribute ``simplest"
to single out such systems from the set of all potential systems associated to equations~\eqref{eqDKfh}.

The spaces of local conservation laws in cases 1--4 of Theorem~\ref{TheoremOnClassificationCLsWideGroup} are one-dimensional and
give rise to the following nonlinear simplest potential systems:
\\[1ex]
\noindent{\bf 1.} $h=1$:  \qquad $v_x=fu$, $v_t=Au_x+\int B$
\\[1ex]
{\bf 2.} $A=1$, $B_u\ne0$, $f=-h(h^{-1})_{xx}$:  \qquad
$v_x=e^t(h^{-1})_{xx} u$, $v_t=e^t(-h^{-1}u_x+(h^{-1})_xu-\int\!\! B).$
\\[1ex]
{\bf 3.} $B=1$, $f=h_x$: \qquad $v_x=e^th_xu$, $v_t=e^t(Au_x+hu).$
\\[1ex]
{\bf 4.} $B=1$, $f=h_x+hx^{-1}$: \qquad $v_x=e^txfu$, $v_t=e^t(xAu_x+xhu-\int\!\! A).$
\\

If the dimension of the space of local conservation laws is greater than one then the situation is more complicated.
Previously, for construction of simplest potential systems only basis conservation laws were used
(e.g., like systems {\bf 1} and {\bf 5} for case 5a of Theorem~\ref{TheoremOnClassificationCLsWideGroup}).
It was just luck that all possible inequivalent potential systems are exhausted for the considered cases.
It happened because arbitrary linear combination of basis conservation laws could be reduced to a subset of the basis elements
by means of a symmetry transformation of the considered system. Here we will show that the basis conservation laws may be
equivalent with respect to group of symmetry transformations, or vice versa, the number of $G^{\sim}$-independent linear combinations
may be grater than the dimension of the space of conservation laws.
Class~\eqref{eqDKfh} is very rich in the sense that it provides us with examples of all these three cases.

Case 5a of Theorem~\ref{TheoremOnClassificationCLsWideGroup} is the classical. The most general form of the simplest system is
\[\textstyle
v_x=(c_1x+c_2)fu,\quad v_t=(c_1x+c_2)Au_x-\int A,
\]
where $c_1$ and $c_2$ are arbitrary constants.
Here and below $(c_1,c_2)\ne(0,0)$. Moreover, without lost of generality, we can assume that the coefficients of linear combinations are determined
up to nonzero factor (in most of cases $c_1^2+c_2^2=1$).

If $c_1=0$ the obtained system
coincides with system {\bf 1}$|_{B=0}$. If $c_1\ne0$ it can be mapped to the system\\[1ex]
\noindent{\bf 5.}  $B=0$: \qquad $v_x=xfu$,\quad $v_t=xAu_x-\int A$\\[1ex]
by means of equivalence transformation of translation in $x$.
We emphasize that in case $f_x\ne0$  this transformation changes the value of arbitrary element $f$.
Therefore, we can reconstruct all potential symmetries of the subclass $B=0$ of class~\eqref{eqDKfh}
applying translations in the space variable to the
arbitrary element~$f$ and to the potential symmetries obtained from systems~{\bf 1} and~{\bf 5}.

Potential system corresponding to case~6 has the form
\[\textstyle
v_x=e^{\mu t}(c_1x+c_2tx-c_2)fu, \quad v_t=e^{\mu t}(c_1x+c_2tx-c_2)(Au_x+hu)-e^{\mu t}c_1\int\!\! A.
\]
In contrast to the previous case, we can simplify the form of the potential system using the symmetry transformations
(translation of time) of the fixed equation of the given form.
More precisely, using the time translation and scaling of potential~$v$ one can reduce the potential system
to one of the following forms:

\noindent{\bf 6.1.} $B=1$, $f=e^{-\mu/x}x^{-3}$, $h=e^{-\mu/x}x^{-1}$, $\mu\in\{0,1\}$:
\begin{gather*}\textstyle
v_x=e^{\mu t}xfu, \quad v_t=e^{\mu t}x(Au_x+hu)-e^{\mu t}\int\!\! A,
\end{gather*}
if $c_2=0$ or

\noindent{\bf 6.2.} $B=1$, $f=e^{-\mu/x}x^{-3}$, $h=e^{-\mu/x}x^{-1}$, $\mu\in\{0,1\}$:
\begin{gather*} \textstyle
v_x=e^{\mu t}(tx-1)fu , \quad v_t=e^{\mu t}(tx-1)(Au_x+hu)-te^{\mu t}\int\!\! A
\end{gather*}
if $c_2\ne0$.

The most interesting from this point of view is the case~7 of Theorem~\ref{TheoremOnClassificationCLsWideGroup}.
It is obvious that the most general form of the associated simplest potential system is
\begin{gather*}
v_x=e^{2\mu t}(x(c_1\cosh t+c_2\sinh t)-c_1\sinh t-c_2\cosh t)fu,\\
v_t=e^{2\mu t}(x(c_1\cosh t+c_2\sinh t)-c_1\sinh t-c_2\cosh t)(Au_x+hu)\\ \textstyle
\phantom{v_x}{}-e^{2\mu t}(c_1\cosh t+c_2\sinh t)\int\!\! A\,).
\end{gather*}
Without lost of generality we can assume $c_1\ge0$. Cases $c_1>|c_2|\ge0$ are equivalent
(under the action of symmetry transformations of translation in time) to the system\\[1ex]
\noindent{\bf 7.1.} $B=1$, $f=|x-1|^{\mu-3/2}|x+1|^{-\mu-3/2}$, $h=|x-1|^{\mu-1/2}|x+1|^{-\mu-1/2}$:
\begin{gather*}\textstyle
v_x=e^{2\mu t}(x\cosh t-\sinh t)fu,\quad v_t=e^{2\mu t}(x\cosh t-\sinh t)(Au_x+hu)-e^{2\mu t}\cosh t\int\!\! A.
\end{gather*}
Similarly cases $|c_2|>c_1\ge0$ fall to the system\\[1ex]
\noindent{\bf 7.2.} $B=1$, $f=|x-1|^{\mu-3/2}|x+1|^{-\mu-3/2}$, $h=|x-1|^{\mu-1/2}|x+1|^{-\mu-1/2}$:
\begin{gather*}\textstyle
v_x=e^{2\mu t}(x\sinh t-\cosh t)fu ,\quad v_t=e^{2\mu t}(x\sinh t-\cosh t)(Au_x+hu)-e^{2\mu t}\sinh t\int\!\! A.
\end{gather*}

At last, we should consider the case of $|c_1|=|c_2|$. Using the scaling transformation of the potential~$v$
we can reduce them either to\\[1ex]
\noindent{\bf 7.3.} $B=1$, $f=|x-1|^{\mu-3/2}|x+1|^{-\mu-3/2}$, $h=|x-1|^{\mu-1/2}|x+1|^{-\mu-1/2}$:
\begin{gather*}\textstyle
v_x=e^{(2\mu+1)t}(x-1)fu,\quad v_t=e^{(2\mu+1)t}(x-1)(Au_x+hu)-e^{(2\mu+1)t}\int\!\! A,
\end{gather*}
if $c_1=c_2$ or to\\[1ex]
\noindent$B=1$, $f=|x-1|^{\mu-3/2}|x+1|^{-\mu-3/2}$, $h=|x-1|^{\mu-1/2}|x+1|^{-\mu-1/2}$:
\begin{gather*}\textstyle
v_x=e^{(2\mu-1)t}(x+1)fu ,\quad v_t=e^{(2\mu-1)t}(x+1)(Au_x+hu)-e^{(2\mu-1)t}\int\!\! A
\end{gather*}
in the case $c_1=-c_2$. Since we consider simultaneously the whole class of equations~\eqref{eqDKfh},
we can apply additionally discrete equivalence transformation of alternating of sign in the set $(t,x,u,v,f,h,A,\mu)$
and reduce the latter system to~{\bf 7.3}. We emphasize once more that if one considers
a separate fixed equation from case~7 of Theorem~\ref{TheoremOnClassificationCLsWideGroup}, he obtains
4 independent potential systems. Considering simultaneously the class of equations we have three $\hat G^{\sim}_1$-independent systems.
Nevertheless, in this case the number of independent potential systems is greater than the number of basis conservation laws.

It is not difficult to show that under the action of symmetry transformation of translation in time
and scaling of potential variable there exist
only one locally inequivalent simplest potential system in case~8 of Theorem~\ref{TheoremOnClassificationCLsWideGroup}:
\\[1ex]
\noindent{\bf 8.} $B=1$, $f=e^{\mu\arctan x}(x^2+1)^{-3/2}$, $ h=e^{\mu\arctan x}(x^2+1)^{-1/2}$:
\begin{gather*}\textstyle
v_x=e^{\mu t}(x\cos t+\sin t)fu ,\quad v_t=e^{\mu t}(x\cos t+\sin t)(Au_x+hu)-e^{\mu t}\cos t\int\!\! A.
\end{gather*}
Thus, the number of independent potential systems is smaller then the number of basis conservation laws.

\begin{theorem}
The complete set of $\hat G^{\sim}_1$-independent simplest potential systems of equations from class~\eqref{eqDKfh} is exhausted by
ones with bold numbers {\bf 1}--{\bf 5}, {\bf 6.1}, {\bf 6.2}, {\bf 7.1}--{\bf 7.3}, {\bf 8}.
\end{theorem}

As one can see, in cases 5a--8 of Theorem~\ref{TheoremOnClassificationCLsWideGroup}
the spaces of local conservation laws are two-dimensional.
It allows us to introduce {\it extended potential systems}~\cite{Popovych&Ivanova2004ConsLawsLanl}
by means of introducing two potentials for each case simultaneously:
\\[1ex]
\noindent{\bf 5$'$.}  $B=0$: \qquad   $v_x=fu$, $v_t=Au_x$, $w_x=xfu$, $w_t=xAu_x-\int\!\! A$.
\\[1ex]
{\bf 6$'$.} $B=1$, $f=e^{-\mu/x}x^{-3}$, $h=e^{-\mu/x}x^{-1}$, $\mu\in\{0,1\}$:
\begin{gather*}\textstyle
v_x=e^{\mu t}xfu, \quad v_t=e^{\mu t}x(Au_x+hu)-e^{\mu t}\int\!\! A
\\ \textstyle
w_x=e^{\mu t}(tx-1)fu , \quad w_t=e^{\mu t}(tx-1)(Au_x+hu)-te^{\mu t}\int\!\! A.
\end{gather*}

\noindent{\bf 7$'$.} $B=1$, $f=|x-1|^{\mu-3/2}|x+1|^{-\mu-3/2}$, $h=|x-1|^{\mu-1/2}|x+1|^{-\mu-1/2}$:
\begin{gather*}\textstyle
v_x=e^{(2\mu+1)t}(x-1)fu,\quad v_t=e^{(2\mu+1)t}(x-1)(Au_x+hu)-e^{(2\mu+1)t}\int\!\! A
\\  \textstyle
w_x=e^{(2\mu-1)t}(x+1)fu ,\quad w_t=e^{(2\mu-1)t}(x+1)(Au_x+hu)-e^{(2\mu-1)t}\int\!\! A.
\end{gather*}

\noindent{\bf 8$'$.} $B=1$, $f=e^{\mu\arctan x}(x^2+1)^{-3/2}$, $ h=e^{\mu\arctan x}(x^2+1)^{-1/2}$:
\begin{gather*}\textstyle
v_x=e^{\mu t}(x\cos t+\sin t)fu ,\quad v_t=e^{\mu t}(x\cos t+\sin t)(Au_x+hu)-e^{\mu t}\cos t\int\!\! A,
\\ \textstyle
w_x=e^{\mu t}(x\sin t-\cos t)fu, \quad w_t=e^{\mu t}(x\sin t-\cos t)(Au_x+hu)-e^{\mu t}\sin t\int\!\! A.
\end{gather*}

Potential symmetries of equations~\eqref{eqDKfh} arising from the above independent potential symmetries
will be investigated in the last part~\cite{Ivanova&Popovych&Sophocleous2006Part4} of this series.

\section{Potential conservation laws\\ of diffusion--convection equations}\label{SectionOnPotConsLawsOfDCEs}

All potential systems of equations~\eqref{eqDKfh} are constructed with usage of local
conservation laws of equations~\eqref{eqDKfh} which were classified with respect to~$\hat G^{\sim}_1$.
Each of these subclasses of conservation laws is equivalent with respect to a subgroup of~$G^{\sim}_1$.
In each case this subgroup has a very simple structure and can be singled out from~\eqref{EquivTransformationsDKfh}
imposing an additional condition $\delta_7=\delta_8=0$. It becomes the usual equivalence group of class~\eqref{eqDKfh}
and is trivially prolonged to the corresponding potentials.
Henceforth we will denote such prolongation as $G^{\sim}_{\rm pr}$.

Let us investigate local conservation laws of potential systems {\bf1}-–{\bf8$'$}, which have the
form
\begin{equation}\label{PotCL}
D_tF(t, x, u_{(r)}, v_{(r)}(,w_{(r)})) + D_xG(t, x, u_{(r)}, v_{(r)}(,w_{(r)})) = 0.
\end{equation}
These laws can be considered as nonlocal (potential) conservation laws of equations from class~\eqref{eqDKfh}.

\begin{lemma}\label{LemmaOnOrderOfPotCLsOfDCEs}
Any conservation law of form~\eqref{PotCL}
for each of systems~{\bf1}-–{\bf8$'$} from Table~1 has the zero order,
i.e., it is equivalent to a conservation law with a conserved density $F$ and a conserved flux $G$
that are independent on the (non-zero order) derivatives of $u$ and potentials $v$ (and $w$).
\end{lemma}

\begin{proof}
Consider any from the systems~{\bf1}-–{\bf8$'$}.
Taking it and its differential consequences into account, we can exclude
dependence of $F$ and $G$ on the all (non-zero order) derivatives of $v$ (and $w$) and the derivatives of $u$
containing differentiation with respect to $t$.
The remain part of the proof is completely similar to the one of Lemma~\ref{LemmaOnOrderOfConsLawsOfDCEs}.
\end{proof}

Systems~{\bf1$|_{B=0}$}, {\bf3}--{\bf7.3} and~{\bf5$'$}--{\bf8$'$} are quit similar: for all of them $B=\const$.
Analyzing conservation laws of them we prove the following statement.

\begin{lemma}\label{LemmaOnExistPotCLs}
Nonlinear equations~\eqref{eqDKfh} with $B=0$ or $B=1$ have nontrivial potential conservation laws only if $A=u^{-2}\!\!\mod G^{\sim}_{\rm pr}$.
\end{lemma}
\begin{proof}
Each of the simplest potential systems~{\bf1$|_{B=0}$}, {\bf3}--{\bf7.3} has the form
\begin{equation}\label{SysPotSysB01GenForm}
\textstyle
v_x=\lambda fu,\quad v_t=\lambda(Au_x+hu)-\lambda_x\int A,
\end{equation}
where $\lambda=\lambda(t,x)$, is a characteristic of a local conservation law of equation~\eqref{eqDKfh} with $B=\const$.
According to Lemma~\ref{LemmaOnOrderOfPotCLsOfDCEs} we look for its conservation laws in form
\[
D_tF(t,x,u,v)+D_xG(t,x,u,v)=0.
\]
First we expand this expression on the solution manifold of~\eqref{SysPotSysB01GenForm}
and decompose it with respect to unconstrained derivatives of $u$.
Coefficients of $u_t$ and $u_x$ give
\[\textstyle
F_u=0,\quad G=-\lambda F_v\int A+\hat G(t,x,v).
\]
Substituting this to the rest of the conservation law we obtain a classifying equation for the components of the conserved vector:
\begin{gather}
\label{ClassifCondConsVecCLB01GenForm}\textstyle
F_t+F_v(\lambda hu-\lambda_x\int A)-(\lambda F_v)_x\int A+\hat G_x-\lambda^2fF_{vv}u\int A+\lambda f\hat G_vu=0.
\end{gather}

If $u\int A\not\in\langle1,u,\int A\rangle$ then
\[
F_{vv}=0, \quad -\lambda_xF_v-(\lambda F_v)_x=0,\quad F_vh+f\hat G_v=0,\quad F_t+\hat G_x=0.
\]
Therefore,
\[
F=F^1(t,x)v+F^0(t,x),\quad \hat G=-\frac hfF^1v+G^0(t,x).
\]
Up to the usual equivalence of conserved vectors we can assume that $F^0(t,x)=G^0(t,x)=0$.

Consider an equivalent conserved vector
\begin{gather*}
\tilde F=F-D_x(\Phi(t,x)v),\quad \tilde G=G+D_t(\Phi(t,x)v),
\end{gather*}
where $\Phi_x(t,x)=F^1$. It is obvious that this conservation law is local for the equation~\eqref{eqDKfh},
and moreover, it is trivial on the solution manifold of the potential system~\eqref{SysPotSysB01GenForm}.

If $u\int A\in\langle1,u,\int A\rangle$, then $u\int A=c_0+c_1u+c_2\int A$. This is an algebraic equation for $u\int A$.
Up to translations of $u$ we can assume that $c_2=0$. Therefore $u\int A=c_0+c_1u$ or $A\in\{1,u^{-2}\}\!\!\mod G^{\sim}_{\rm pr}$.
Since we consider only nonlinear equations, $A=u^{-2}$ is the only case where potential conservation laws can exist.

Consider now the general potential systems~{\bf5$'$}--{\bf8$'$} on two potentials. Each of them has the form
\begin{equation}\label{SysGenPotSysB01GenForm}
\textstyle
v^i_x=\lambda^i fu,\quad v^i_t=\lambda^i(Au_x+hu)-\lambda^i_x\int A,\quad i=1,2,
\end{equation}
where $\lambda^i=\lambda^i(t,x)$ are linearly independent characteristics of local conservation laws.
Similarly to the case of simplest potential system we can prove that an arbitrary conservation law of system~\eqref{SysGenPotSysB01GenForm}
has the form
\[
D_tF(t,x,v^1,v^2)+D_xG(t,x,u,v^1,v^2)=0,
\]
where the flux~$G=-\lambda^iF_{v^i}\int A+\hat G(t,x,v^1,v^2)$, and the classifying equation is the following
\begin{gather}
\label{ClassifCondGenConsVecCLB01GenForm}\textstyle
F_t+\hat G_x+u\lambda^i(hF_{v^i}+f\hat G_{v^i})-\int A(\lambda^i_xF_{v^i}+(\lambda ^iF_{v^i})_x)-\lambda^i\lambda^jfF_{v^iv^j}u\int A=0.
\end{gather}
Now, if $u\int A\not\in\langle1,u,\int A\rangle$
(in other words, $A\ne1,u^{-2}\!\!\mod G^{\sim}_{\rm pr}$) we can split~\eqref{ClassifCondConsVecCLB01GenForm}
with respect to the linearly independent functions of~$u$ and obtain
\[
F_t+\hat G_x=0,\quad \lambda^i(hF_{v^i}+f\hat G_{v^i})=0,\quad \lambda^i_xF_{v^i}+(\lambda ^iF_{v^i})_x=0,\quad \lambda^i\lambda^jfF_{v^iv^j}=0.
\]
The latter system implies, in particular, that
\begin{gather*}
F=F^1(t,x,\omega)\theta+F^0(t,x,\omega),\quad \hat G=-\frac hfF^1\theta+G^0(t,x,\omega),\\
2F^1\lambda^j\lambda^j_x+\lambda^j\lambda^jF^1 _\omega(\lambda^1_xv^2-\lambda^2_xv^1)
=-F^1 _\omega(\lambda^1\lambda^2_x-\lambda^2\lambda^1_x)\theta-F^1\lambda^1\lambda^1_x-F^0_\omega(\lambda^1\lambda^2_x-\lambda^2\lambda^1_x),
\end{gather*}
where $\omega=\lambda^1v^2-\lambda^2v^1$, $\theta=\lambda^1v^1+\lambda^2v^2$.
Decomposing $(\lambda^1_xv^2-\lambda^2_xv^1)$ with respect to $\omega$, $\theta$ and collecting the coefficients of~$\theta$ in the last equation
we obtain $2F^1 _\omega(\lambda^1\lambda^2_x-\lambda^2\lambda^1_x)=0$.
One can easily check that for each of the systems~{\bf5$'$}--{\bf8$'$} $\lambda^1\lambda^2_x-\lambda^2\lambda^1_x\ne0$, therefore, $F^1 _\omega=0$
which implies immediately $F^0_{\omega\omega}=0$.
Thus, we get that the conserved density and, consequently, the conserved flux are linear with respect to the potentials.
Therefore, in view of~\cite{Bluman&Cheviakov&Ivanova2006,Popovych2007Comment}, the conservation law is local for equation~\eqref{eqDKfh}.

The lemma is completely proved.
\end{proof}

In view of Lemma~\ref{LemmaOnExistPotCLs}, in order to construct potential conservation laws of equations~\eqref{eqDKfh},
it is enough to investigate conservation laws of the potential systems {\bf1}, {\bf2} and cases $A=u^{-2}$ of systems~{\bf3}--{\bf8$'$}.
The chain of lemmas below describes completely the set of potential conservation laws of equations~\eqref{eqDKfh}.

\begin{lemma}\label{LemmaOnPotCLsSys1}
The list of $G^{\sim}_{\rm pr}$-inequivalent linearly independent conservation laws of system~{\bf1} is exhausted by the following ones
\begin{gather*}\textstyle
1.\quad \forall A,\ \int B=u\int A,\ f=1:\qquad D_t(e^v)+D_x(e^v\int A)=0,\\
2.\quad A=u^{-2},\ B=0,\ f=1:\qquad D_t(\sigma)+D_x(\sigma_{v}u^{-1})=0,\\
3.\quad A=1,\ B=2u,\ f=1:\qquad D_t(\alpha e^v)+D_x(\alpha_xe^v-\alpha u e^v)=0,
\end{gather*}
where $\alpha=\alpha(t,x)$ and $\sigma=\sigma(t,v)$ are arbitrary solutions of the backward linear heat equations $\alpha_t+\alpha_{xx}=0$
and $\sigma_t+\sigma_{vv}=0$, correspondingly.
\end{lemma}
\begin{proof}
In view of Lemma~\ref{LemmaOnOrderOfPotCLsOfDCEs} we look for conservation laws of system~{\bf1} in form
\[
D_tF(t,x,u,v)+D_xG(t,x,u,v)=0.
\]
Expanding the total derivatives in the expression for the conservation law
and decomposing it with respect to unconstrained derivatives of $u$ we get
\begin{gather*}\textstyle
F_u=0,\quad G=-F_v\int A+\hat G(t,x,v),\\ \textstyle
F_t+F_v\int B-F_{vx}\int A+\hat G_x-fF_{vv}u\int A+\lambda f\hat G_vu=0.
\end{gather*}
To integrate the latter equation one should consider separately three cases:
\begin{itemize}
    \item $1,u,\int A, u\int A,\int B$ are linearly independent;
    \item $1,u,\int A, u\int A$ are linearly independent and $1,u,\int A, u\int A,\int B$ are linearly dependent;
    \item $u\int A\in\langle1,u,\int A\rangle$.
\end{itemize}
All conservation laws in the first case are trivial. In the second case we have
$
\int B=c_0+c_1u+(c_2u+c_3)\int A.
$
Substituting this into the classifying condition and solving it up to $G^{\sim}_{\rm pr}$ we obtain case 1 of the list of the lemma statement.
At last, if $u\int A\in\langle1,u,\int A\rangle$, then similarly to Lemma~\ref{LemmaOnExistPotCLs}, $A\in\{1,u^{-2}\}\!\!\mod G^{\sim}_{\rm pr}$
that leads after some obvious calculations to cases 2 and 3 of the lemma statement.
\end{proof}

\begin{lemma}\label{LemmaOnPotCLsSys2}
System {\bf2} has no nontrivial conservation laws.
\end{lemma}
\begin{proof}
The proof is completely similar to the proof of Lemma~\ref{LemmaOnPotCLsSys1}.
\end{proof}

\begin{lemma}\label{LemmaOnPotCLsSys34678}
Systems {\bf3}, {\bf4}, {\bf6.1}--{\bf8} have no nontrivial conservation laws.
There exists exactly one nontrivial conservation law for system {\bf5}:
\[
A=u^{-2},\ B=0,\ f=x^{-2}:\qquad D_t(x^{-2}\sigma)+D_x(x^{-1}\sigma_{v}u^{-1})=0,
\]
where $\sigma=\sigma(t,v)$ runs the solution set of the backward heat equation $\sigma_t+\sigma_{vv}=0$.
\end{lemma}
\begin{proof}
Each of the systems~{\bf3}--{\bf8} has the form~\eqref{SysPotSysB01GenForm}.
The conserved density and flux of the conservation law has the form
\[\textstyle
F=F(t,x,v),\quad G=-\lambda F_v\int A+\hat G(t,x,v).
\]
In view of Lemma~\ref{LemmaOnExistPotCLs} it is enough to consider case $A=u^{-2}$ only.
Then~\eqref{ClassifCondConsVecCLB01GenForm} implies
\[
\textstyle
\lambda_xF_v+(\lambda F_v)_x=0,\quad \lambda hF_v+\lambda f\hat G_v=0,\quad F_t+\hat G_x+\lambda^2fF_{vv}=0.
\]
Therefore, up to the equivalence relation of the conserved vectors
$F=\lambda^{-2}\sigma$, $\hat G=-\frac h{f}\lambda^{-2}\sigma$,
where $\sigma=\sigma(t,v)$ satisfies the classifying equation
\[
(\lambda^{-2}\sigma)_t-\left(\frac h{f}\lambda^{-2}\right)_x\sigma+f\sigma_{vv}=0.
\]
Integrating this equation for the given values of characteristics~$\lambda$ implies the lemma statement.
\end{proof}

\begin{lemma}\label{LemmaOnPotCLsSys5'6'7'8'}
The list of $G^{\sim}_{\rm pr}$-inequivalent linearly independent conservation laws of system~{\bf5$'$} is exhausted by the following ones
\begin{gather*}\textstyle
1.\quad A=u^{-2},\ B=0,\ f=1:\qquad D_t(\sigma^1)+D_x(\sigma^1_{v^1}u^{-1})=0,\\
2.\quad A=u^{-2},\ B=0,\ f=x^{-2}:\qquad D_t(x^{-2}\sigma^2)+D_x(x^{-1}\sigma^2_{v^2}u^{-1})=0,,
\end{gather*}
$\sigma^i=\sigma^i(t,v^i)$, $i=1,2$ run the solution sets of the backward heat equations $\sigma^i_t+\sigma^i_{v^iv^i}=0$ correspondingly.
Systems~{\bf6$'$}--{\bf8$'$}  have no nontrivial conservation laws.
\end{lemma}
\begin{note}
These potential conservation laws of equations~\eqref{eqDKfh} indeed coincide with conservation laws of the simplest potential systems~{\bf1} and~{\bf5}.
\end{note}
\begin{proof}
Substituting $A=u^{-2}$ into the classifying condition~\eqref{ClassifCondConsVecCLB01GenForm}
and splitting it with respect to the linearly independent functions of~$u$ we obtain the following system of determining equations
for the coefficients of the conserved vectors of the conservation laws of systems~\eqref{SysGenPotSysB01GenForm}:
\begin{gather}\label{SysDetEqConsVecGenPonSysGenFormu-2}
\lambda^i(hF_{v^i}+f\hat G_{v^i})=0,\quad \lambda^i_xF_{v^i}+(\lambda ^iF_{v^i})_x=0,\quad
F_t+\hat G_x+\lambda^i\lambda^jfF_{v^iv^j}=0.
\end{gather}
From the first equation of~\eqref{SysDetEqConsVecGenPonSysGenFormu-2} we obtain that $\hat G=-\frac hfF+H(t,x,\omega)$.
Here, as in Lemma~\ref{LemmaOnExistPotCLs}, $\omega=\lambda^1v^2-\lambda^2v^1$, $\theta=\lambda^1v^1+\lambda^2v^2$.

It follows from~\eqref{SysGenPotSysB01GenForm} that
\begin{gather*}
\lambda^1v^2_x-\lambda^2v^1_x=0=\omega_x-(\lambda^1_xv^2-\lambda^2_xv^1),\\
\lambda^1v^2_t-\lambda^2v^1_t=(\lambda^1\lambda^2_x-\lambda^2\lambda^1_x)u^{-1}=\omega_t-(\lambda^1_tv^2-\lambda^2_tv^1).
\end{gather*}
Decomposing $\lambda^1_xv^2-\lambda^2_xv^1$ and $\lambda^1_tv^2-\lambda^2_tv^1$ with respect to $(\omega,\theta)$ from the latter equations we get
\begin{gather*}
\omega_x=\frac{\lambda^k_x\lambda^k}{\lambda^j\lambda^j}\omega-\frac{\lambda^1\lambda^2_x-\lambda^2\lambda^1_x}{\lambda^j\lambda^j}\theta,\qquad
\omega_t=\frac{\lambda^k_t\lambda^k}{\lambda^j\lambda^j}\omega-\frac{\lambda^1\lambda^2_t-\lambda^2\lambda^1_t}{\lambda^j\lambda^j}\theta+
\frac{\lambda^1\lambda^2_x-\lambda^2\lambda^1_x}{u}.
\end{gather*}
These equalities allow us to find the simplest form of conserved vectors corresponding to the conservation law. For now we have
$F=F(t,x,\omega,\theta)$, $G=\lambda^j\lambda^jF_{\theta}u^{-1}-\frac hfF+H(t,x,\omega)$.
An equivalent conserved vector can be written as
\begin{gather*}
\tilde F=F+D_x\Phi(t,x,\omega)=F+\Phi_x+\Phi_\omega\omega_x,\\
\tilde G=G-D_t\Phi(t,x,\omega)=\lambda^j\lambda^jF_{\theta}u^{-1}-\frac hfF+H-\Phi_t-\Phi_\omega\omega_t=\\
%\lambda^j\lambda^j\tilde F_{\theta}u^{-1}-\frac hf\tilde F+\frac hf(\Phi_x+\Phi_\omega\omega_x)+H-\Phi_t-
%\Phi_\omega\left(\frac{\lambda^k_t\lambda^k}{\lambda^j\lambda^j}\omega
%-\frac{\lambda^1\lambda^2_t-\lambda^2\lambda^1_t}{\lambda^j\lambda^j}\theta\right)=
%\\
\lambda^j\lambda^j\tilde F_{\theta}u^{-1}-\frac hf\tilde F+
\Phi_\omega\theta\left(-\frac hf\frac{\lambda^1\lambda^2_x-\lambda^2\lambda^1_x}{\lambda^j\lambda^j}
                       +\frac{\lambda^1\lambda^2_t-\lambda^2\lambda^1_t}{\lambda^j\lambda^j}\right)
\\
+H+\frac hf\left(\Phi_x+\Phi_\omega\frac{\lambda^k_x\lambda^k}{\lambda^j\lambda^j}\omega\right)
-\Phi_t-\Phi_\omega\frac{\lambda^k_t\lambda^k}{\lambda^j\lambda^j}\omega.
\end{gather*}
For all of the characteristics corresponding to systems~{\bf5$'$}--{\bf8$'$},
coefficients of $\Phi_\omega\theta$ in the expression for $\tilde G$ is vanishing.
Therefore, choosing $\Phi$ in such way that
\[
H+\frac hf\left(\Phi_x+\Phi_\omega\frac{\lambda^k_x\lambda^k}{\lambda^j\lambda^j}\omega\right)
-\Phi_t-\Phi_\omega\frac{\lambda^k_t\lambda^k}{\lambda^j\lambda^j}\omega=0,
\]
we obtain that
$
\tilde G=\lambda^j\lambda^j\tilde F_{\theta}u^{-1}-\frac hf\tilde F,
$
i.e., $\tilde H=0$. That is why, without lost of generality, we assume that
\[
F=F(t,x,v^1,v^2),\quad G=\lambda^i\lambda^iF_{\theta}u^{-1}-\frac hfF=\lambda^iF_{v^i}u^{-1}-\frac hfF,
\]
where the conserved density $F$ satisfies the following system:
\[
\lambda^i_xF_{v^i}+(\lambda ^iF_{v^i})_x=0,\quad
\lambda^i\lambda^jF_{v^iv^j}+\frac{F_t}f-\frac1f\left(\frac hfF\right)_x=0.
\]
Taking into account its nontrivial differential consequence
\[
\left(\frac{F_t}f-\frac1f\left(\frac hfF\right)_x\right)_x=0,
\]
we can write the compatible system of determining equations for the conserved density~$F$:
\begin{gather} \nonumber
\frac{F_t}f-\frac1f\left(\frac hfF\right)_x=K(t,v^1,v^2),
\quad \lambda^i\lambda^jF_{v^iv^j}+K=0,\\ \label{SysDetEqSimplifiedConsVecGenPonSysGenFormu-2}
\lambda^i_xF_{v^i}+(\lambda ^iF_{v^i})_x=0.
\end{gather}

System~{\bf5$'$} is distinguished from the set of the systems~\eqref{SysGenPotSysB01GenForm} since we can assume $h=0$.
Let us investigate it separately.
Substituting $h=0$ and the characteristics $\lambda^1=1$, $\lambda^2=x$ into~\eqref{SysDetEqSimplifiedConsVecGenPonSysGenFormu-2}
and taking its differential consequence, one can obtain
\begin{equation}\label{eqForFinSys5'}
f_xK_{v^1}+(xf_x+2f)K_{v^2}=0.
\end{equation}
There exist three different cases of integration of equation~\eqref{eqForFinSys5'}: $f$ is arbitrary function such that $f_x$ and $xf_x+2f$
are linearly independent, $f_x=0$ and $xf_x+2f=cf_x$, where $c$ is a constant.

If $f_x$ and $xf_x+2f$ are linearly independent then $K_{v^1}=K_{v^2}=0$. Then, it follows from~\eqref{SysDetEqSimplifiedConsVecGenPonSysGenFormu-2}
that $K_t=0$. Therefore, $K$ is a constant, and without loss of generality we can assume $K\in\{0,1\}$.
If $K=1$, then~\eqref{SysDetEqSimplifiedConsVecGenPonSysGenFormu-2} implies
\begin{gather*}
F=ft+M(x,v^1,v^2),\quad M=-\frac12\theta^2+M^1(x,\omega)\theta+M^0(x,\omega),\\
(\lambda^iM_{v^i})_x+\lambda^i_xM_{v^i}=0.
\end{gather*}
Collecting coefficients of $\theta$ in the last equation we get $-x\theta=0$, that is contradiction.
Therefore $K=0$. Then $F_{v^i}=G_{v^i}=0$, and the conservation law is local.

If $f_x=0$ then $f=1\!\!\mod\hat G^{\sim}$, $K_{v^2}=0$. Then, from~\eqref{SysDetEqSimplifiedConsVecGenPonSysGenFormu-2} we obtain
$K_{v^1v^1}+K_t=0$ and $F=\sigma^1(t,v^1)+M(x,v^1,v^2)$, where $K=\sigma^1_t=-\sigma^1_{v^1v^1}$.
Substituting it into the second equation of~\eqref{SysDetEqSimplifiedConsVecGenPonSysGenFormu-2} leads to the equation for~$M$ of form
$\lambda^i\lambda^jM_{v^iv^j}=0$. Therefore, $M=M^1(x,\omega)\theta+M^0(x,\omega)$, and $M^1$, $M^0$ satisfy the equation
\[
M^0_\omega-x\omega M^1_\omega+M^1_x(x^2+1)+3xM^1=0.
\]
Considering equivalent conservation law with the conserves vector $\tilde F=F-D_x(N\theta)$, $\tilde G=G+D_t(N\theta)$, where
$N_\omega(x,\omega)=M^1$ one can show that $\tilde N=0$. Thus we get the conservation law from case~1 of the lemma statement.

If $xf_x+2f=cf_x$, then up to translations of~$x$ we can assume that $c=0$ and $f=x^{-2}$.
The local transformation $\tilde t=t$, $\tilde x=-1/x$, $\tilde u=xu$, $\tilde v^1=v^2$, $\tilde v^2=v^1$
reduces this problem to the previous case $f=1$.

Case $h\ne0$ (systems~{\bf6$'$}--{\bf8$'$}) can be investigated in a similar manner.
At first, we find the differential consequence of system~\eqref{SysDetEqSimplifiedConsVecGenPonSysGenFormu-2}:
\begin{gather*}
(\lambda^if_x+2\lambda^i_xf)K_{v^i}
+F_{v^i}\left(\Bigl(\frac hf\Bigr)_{xx}\lambda^i+\frac2\lambda^i_x\Bigl(\frac hf\Bigr)_{x}+2\lambda^i_{tx}-2\lambda^i_{xx}\frac hf\right)\\
+F_{v^ix}\left(2\lambda^i\Bigl(\frac hf\Bigr)_{x}-\lambda^i_{x}\frac hf+\lambda^i_t\right)=0.
\end{gather*}
Considering the given sets of characteristics for each of the systems {\bf6$'$}--{\bf8$'$} we obtain
\[
(\lambda^if_x+2\lambda^i_xf)K_{v^i}=0,
\]
where $(\lambda^if_x+2\lambda^i_xf)$ are linearly independent. Therefore, $K_{v^i}=K_t=0$. Absolutely similarly to case $h=0$ we get
that the conservation law is local.

The lemma is proved.
\end{proof}

As one can see, equation~\eqref{eqDKfh} admits nontrivial potential conservation laws of the first level iff it is linearizable
or satisfies the conditions $f=h=1$, $\int B=u\int A$.
It is shown in~\cite{Popovych&Ivanova2004ConsLawsLanl} that such equations do not have potential conservation laws of the second level.
Therefore, the following lemma is true.

\begin{lemma}\label{LemmaOnSecLevelPotCLs}
Any equation of form~\eqref{eqDKfh} has no nontrivial second level potential conservation laws.
\end{lemma}

This statement completes investigation of potential conservation laws of the variable coefficient diffusion--convection equations~\eqref{eqDKfh}.
Summarizing the above results, we can formulate the following theorem.

\begin{theorem}
A complete set of $G^{\sim}_{\rm pr}$-inequivalent potential conservation laws of nonlinear equations~\eqref{eqDKfh} consists of the following ones:
\begin{gather*}\textstyle
1.\quad \forall A,\ \int B=u\int A,\ f=1:\qquad D_t(e^v)+D_x(e^v\int A)=0,\\
2.\quad A=u^{-2},\ B=0,\ f=1:\qquad D_t(\sigma)+D_x(\sigma_{v}u^{-1})=0,\\
3.\quad A=1,\ B=2u,\ f=1:\qquad D_t(\alpha e^v)+D_x(\alpha_xe^v-\alpha u e^v)=0,
\end{gather*}
where the potential variable $v(t,x)$ satisfies the potential system $v_x=u$, $v_t=Au_x+\int B$,
\begin{gather*}\textstyle
4.\quad A=u^{-2},\ B=0,\ f=x^{-2}:\qquad D_t(x^{-2}\sigma)+D_x(x^{-1}\sigma_{v}u^{-1})=0,
\end{gather*}
where $v_x=x^{-1}u$, $v_t=xu^{-2}u_x+u^{-1}$.
Here $\alpha=\alpha(t,x)$ and $\sigma=\sigma(t,v)$ are arbitrary solutions of the backward heat equations $\alpha_t+\alpha_{xx}=0$
and $\sigma_t+\sigma_{vv}=0$, correspondingly,
\end{theorem}

\section{Conclusion}

In this paper we generalized the notion of equivalence of conservation laws with respect to groups of transformations and adduced
description of generating sets of conservation laws.
This framework is applied to investigation of conservation laws of variable coefficient diffusion--convection equations.
We classified the local conservation laws of equations~\eqref{eqDKfgh} both up to the usual group of equivalence transformations
and up to the extended group containing nonlocal (with respect to the arbitrary elements) transformations.
Usage of equivalence of conservation laws with respect to the extended equivalence group and correct choice
of gauge coefficients of equations allow us to obtain clear formulation of the final results.

The notions of contractions conservation laws and ones of characteristics of conservation laws are introduced.
We presented some examples of such contractions in class~\eqref{eqDKfgh}.
The problem of finding all possible contractions of equations and conservation laws in class~\eqref{eqDKfgh} remains open.

The proposed view on equivalence of conservation laws give rise to a generalization of the theory of potential symmetries.
Namely we showed how to find {\em all possible} inequivalent potential systems associated to the given system of differential equations.
We classified such potential systems and potential conservation laws for equations~\eqref{eqDKfgh}.

Using this classification one can find all potential symmetries of class~\eqref{eqDKfgh}.
This is a subject of the sequel part~\cite{Ivanova&Popovych&Sophocleous2006Part4} of this series of papers.

\subsection*{Acknowledgements}
NMI and ROP express their gratitude to the hospitality shown by University of Cyprus
during their visits to the University.
Research of NMI was supported by the Erwin Schr\"odinger Institute for Mathematical Physics (Vienna, Austria) in form of Junior Fellowship
and by the grant of the President of Ukraine for young scientists (project number GP/F11/0061).
Research of ROP was supported by Austrian Science Fund (FWF), Lise Meitner project M923-N13.

\end{document}